\documentclass[journal]{IEEEtran}

\usepackage{lineno} 

\usepackage{amssymb} 
\usepackage{amsmath,amsfonts}
\usepackage{array}
\usepackage[caption=false,font=scriptsize,labelfont=sf,textfont=sf]{subfig}
\usepackage{textcomp}
\usepackage{stfloats}
\usepackage{url}
\usepackage{verbatim}
\usepackage{graphicx}
\usepackage{cite}

\usepackage[bookmarksnumbered=true]{hyperref}
\hypersetup{
hidelinks,
colorlinks=true,
linkcolor=blue,
citecolor=blue,
urlcolor=black
}
\usepackage[linesnumbered,ruled,vlined]{algorithm2e}
\usepackage{bbm} 

\usepackage{amsthm} 

\newtheorem{proposition}{Proposition}

\usepackage{hyperref} 
\usepackage{xcolor}

\usepackage{enumerate} 

\begin{document}

\newcounter{probcounter}
\newcommand{\refporbcounter}[1]{\refstepcounter{probcounter}\theprobcounter \label{#1}} 


\title{Energy-Efficient RSMA-enabled Low-altitude MEC Optimization Via Generative AI-enhanced Deep Reinforcement Learning}

\author{
Xudong Wang, 
Hongyang Du,
Lei Feng,~\IEEEmembership{Member,~IEEE,}
and Kaibin Huang,~\IEEEmembership{Fellow,~IEEE}

\thanks{\textit{Corresponding author: Lei Feng.}}
\thanks{Xudong Wang and Lei Feng are with the State Key Laboratory of Networking and Switching Technology, Beijing University of Posts and Telecommunications, Beijing 100876, China (e-mail: xdwang@bupt.edu.cn; fenglei@bupt.edu.cn).}
\thanks{Hongyang Du and Kaibin Huang are with the Department of Electrical and Electronic Engineering, University of Hong Kong, Pok Fu Lam, Hong Kong SAR, China (e-mail: duhy@eee.hku.hk; huangkb@hku.hk).}
}



\maketitle

\begin{abstract}
The growing demand for low-latency computing in 6G is driving the use of UAV-based low-altitude mobile edge computing (MEC) systems. However, limited spectrum often leads to severe uplink interference among ground terminals (GTs).
In this paper, we investigate a rate-splitting multiple access (RSMA)-enabled low-altitude MEC system, where a UAV-based edge server assists multiple GTs in concurrently offloading their tasks over a shared uplink. We formulate a joint optimization problem involving the UAV 3D trajectory, RSMA decoding order, task offloading decisions, and resource allocation, aiming to mitigate multi-user interference and maximize energy efficiency. Given the high dimensionality, non-convex nature, and dynamic characteristics of this optimization problem, we propose a generative AI-enhanced deep reinforcement learning (DRL) framework to solve it efficiently. Specifically, we embed a diffusion model into the actor network to generate high-quality action samples, improving exploration in hybrid action spaces and avoiding local optima. In addition, a priority-based RSMA decoding strategy is designed to facilitate efficient successive interference cancellation with low complexity. Simulation results demonstrate that the proposed method for low-altitude MEC systems outperforms baseline methods, and that integrating GDM with RSMA can achieve significantly improved energy efficiency performance.

\end{abstract}


\section{Introduction}
In the sixth generation (6G) era, the proliferation of Internet of Things (IoT) devices and new immersive applications, such as augmented/virtual reality and holographic telepresence, is driving an unprecedented demand for computational resources at the network edge~\cite{8792135}. These latency-sensitive and computation-intensive services exceed the capabilities of onboard processors in battery-powered ground terminals (GTs), resulting in performance degradation when tasks are executed locally~\cite{8391395}. Mobile edge computing (MEC) offers a viable solution by offloading tasks to proximate edge servers, reducing latency and improving efficiency~\cite{8391395}. However, terrestrial BS-centric MEC is constrained by coverage and capacity, especially in remote or infrastructure-limited environments~\cite{9366869}, or when BSs become overloaded or fail to support edge users. These limitations motivate the exploration of airborne MEC, wherein aerial platforms augment the edge computing infrastructure beyond what ground-based MEC can provide.

In recent years, the low-altitude economy (LAE) has expanded swiftly, with unmanned aerial vehicles (UAVs) emerging as key components in the evolution of mobile communication networks~\cite{10955337,zhao2025generative}. Dynamically deployed UAVs can form LAE networks that augment terrestrial infrastructure with additional communication and computing capabilities. When equipped with computing units, UAVs serve as airborne MEC servers, enabling low-altitude MEC systems that bring computation offloading closer to users~\cite{7932157}. Such UAV-assisted MEC networks improve coverage and service reliability while significantly reducing end-to-end latency and GT energy consumption, as users no longer need to transmit to distant base stations. Nevertheless, supporting multiple offloading users via a shared UAV radio link introduces substantial multi-user interference. When many users simultaneously offload tasks over limited spectrum resources, interference among their uplink transmissions becomes severe and can drastically bottleneck the system performance, especially under strict bandwidth constraints.

Rate-splitting multiple access (RSMA) has emerged as a promising strategy to manage interference and improve spectral efficiency in multi-user networks~\cite{9831440}. RSMA separates a user's transmission into private and common components, allowing for a combination of interference decoding and noise handling strategies~\cite{10813419}. Unlike uplink Non-Orthogonal Multiple Access (NOMA) with fixed user grouping and power allocation, RSMA offers greater flexibility through adjustable message splits and decoding orders, enhancing spectral efficiency. This makes RSMA well-suited for low-altitude UAV MEC, where multiple GTs can offload concurrently over the shared spectrum without orthogonalization~\cite{10059199}. While existing RSMA studies focus on downlink, extending it to uplink UAV-assisted MEC systems offers new potential for interference mitigation. However, this introduces complex design challenges, requiring joint optimization of UAV trajectory, offloading decisions, RSMA decoding order, and power allocation. Addressing these tightly coupled variables remains a significant challenge.

Conventional optimization techniques, such as convex optimization and heuristic algorithms, are widely used for joint resource and trajectory design but often suffer from high complexity and local optima~\cite{10759093}. Deep reinforcement learning (DRL) offers an online alternative, capable of learning near-optimal policies for joint scheduling and control in high-dimensional environments~\cite{10192095}. However, DRL faces challenges in continuous action spaces, including poor sample efficiency and convergence instability~\cite{you2025reacritic}. To address these drawbacks, some studies incorporate generative diffusion models (GDMs), which is a class of generative AI (GenAI) techniques adept at modeling complex distributions through iterative sampling~\cite{10529221,wang2024generative,10812969,you2025dress}. By iteratively refining random samples to match an intended distribution, diffusion models excel at generating high-quality solutions even in large search spaces~\cite{10812969}. This motivates our proposed RSMA-assisted low-altitude MEC framework, combining generative AI-enhanced DRL for improved energy efficiency. 

\subsection{Related Works}

Low-altitude MEC systems are characterized by dynamic node locations, limited computational resources, and a strong dependence on communication link quality, making UAV trajectory planning and computation offloading particularly challenging. As a result, recent studies have investigated the joint optimization of UAV trajectories and resource allocation in low-altitude MEC systems~\cite{9814972,9384267,9672750,10546306}. 
For example, the authors in~\cite{9814972} designed a joint task offloading and trajectory control scheme for UAV-assisted energy-harvesting MEC systems to minimize propulsion energy while maintaining queue stability. The authors in~\cite{9384267,9672750} proposed joint optimization of communication resources and UAV trajectories to improve offloading throughput or reduce energy consumption under quality of service (QoS) and energy constraints. In~\cite{10546306}, a bi-objective ant colony algorithm addressed delay and control cost minimization problem.
These studies provide valuable insights into low-altitude edge resource management. However, transmission frameworks based on orthogonal multiple access (OMA) struggle to accommodate the explosive growth in network density, particularly in spectrum-constrained low-altitude networks, often leading to degraded performance in interference management and task completion rates.

Leveraging its capability to mitigate multi-user interference and enhance spectral efficiency, RSMA has been incorporated into low-altitude networks to improve communication performance under interference-limited conditions~\cite{10411132,9585491,9151975,10742902,kim2023energy}. In~\cite{10411132}, an alternating optimization strategy for RSMA-enabled low-altitude networks was proposed to maximize system throughput while ensuring user rate requirements. The work in~\cite{9585491} demonstrated that RSMA significantly improves physical-layer security in a dual-UAV downlink scenario. Additionally, RSMA has been introduced into low-altitude MEC systems to further improve offloading performance. Specifically, the authors in~\cite{9151975} designed a UAV relay architecture and propose a data-stream splitting scheme that minimizes UAV energy consumption under task latency constraints. In~\cite{10742902}, a joint passive reflecting units and dynamic rate-splitting optimization framework was proposed for uplink RSMA-enabled UAV-MEC systems to minimize total UAV power consumption. The authors in~\cite{kim2023energy} applied RSMA to a UAV-MEC system assisted by a reconfigurable intelligent surface (RIS) with wireless power transfer and observed that the RSMA-based scheme achieves more reliable and efficient task offloading in the presence of strong co-channel interference. Dynamic low-altitude MEC systems must adapt to time-varying channels and fluctuating interference levels, but traditional iterative optimization methods face scalability challenges when dealing with highly coupled decision variables in such dynamic environments.

Data-driven optimization methods, particularly DRL, have emerged as powerful solutions for complex optimization problems in low-altitude MEC systems due to their dynamic adaptability and global exploration capabilities. For example, the authors in~\cite{9968236} employed DRL to jointly adjust UAV flight paths and offloading decisions, effectively reducing the system's energy consumption and task latency compared to heuristic baselines. To further improve solution quality and avoid local optima, GDMs have been introduced to augment DRL in wireless network optimization~\cite{10841385,zhao2025temporal}. By leveraging the diffusion-based generative process, agents can enhance their exploration of complex state spaces and refine reward evaluation during training~\cite{zhang2025improve,10852212,10892236,10472660}. For instance, the authors in~\cite{zhang2025improve} introduced a diffusion-assisted DRL framework that generates diverse state samples and more informative reward signals, achieving faster convergence and lower computational overhead compared to standard DRL in a dynamic resource allocation task. Similarly, in~\cite{10472660}, a diffusion-enhanced DRL framework was applied to a wireless sensing-guided edge network to derive optimal pricing strategies, thereby maximizing user utility. Therefore, given the generative capability, integrating diffusion models with DRL to address joint trajectory and resource allocation in RSMA-enabled low-altitude MEC systems presents a promising research direction.

\subsection{Motivations and Contributions}

This paper presents a joint optimization framework for uplink RSMA-enabled low-altitude MEC systems to maximize energy efficiency. We proposed a GenAI-enhanced DRL algorithm to jointly optimize UAV trajectory, user offloading, RSMA decoding order, and power allocation. By embedding a GDM generator into the DRL agent, the proposed algorithm effectively explores complex solution spaces and mitigates local optima issues common in conventional approaches. The main contributions are summarized as follows:
\begin{itemize}
    \item We develop a uplink RSMA-assisted low-altitude MEC system, where a UAV equipped with an mobile edge server provides computing services to multiple GTs, and the task data of each GT is split into multiple sub-messages that are concurrently transmitted with different powers. Then we formulate a novel non-convex optimization problem that maximizes energy efficiency by jointly optimizing the UAV's flight trajectory, user offloading decisions, and RSMA communication parameters.
    \item We propose a generative AI-enhanced deep reinforcement learning framework, termed GDRS, to efficiently solve the high-dimensional, non-convex optimization problem in RSMA-enabled UAV-MEC systems. The problem is decomposed into a decoding-order sub-problem and a continuous control problem. To reduce the decoding complexity, we derive a lightweight priority-based RSMA strategy that dynamically orders sub-messages based on channel quality and user QoS demands, enabling interference cancellation with low overhead.
    \item To overcome the challenges of poor exploration and local optima, we embed a generative diffusion model into the policy network. This integration allows the agent to generate structured candidate actions, including UAV trajectories and resource allocation decisions, leading to more stable convergence and higher policy quality across dynamic and coupled state-action spaces.
    \item We carried out a comprehensive set of simulations to validate the proposed GDRS and to assess the scalability of the RSMA framework. The results show that our RSMA-enabled framework achieves higher throughput and energy efficiency compared to NOMA-based schemes. Moreover, the GenAI-enhanced DRL method outperforms conventional DRL approaches in terms of reward, highlighting the superior exploration capability of the diffusion model.
\end{itemize}

The rest of this paper is organized as follows. Section~\ref{sec_model} introduces the system model and formally defines the optimization problem. In Section~\ref{algorithm}, we elaborate on the algorithmic framework developed to address the problem. Section~\ref{sec_simulation} showcases simulation results and provides an in-depth discussion of the numerical findings. Section~\ref{sec_conclusion} highlights the key contributions and draws the final conclusions of the paper.


\section{System Model And Problem Formulation}\label{sec_model}

As illustrated in Fig.~\ref{fig:system_model}, we investigate an RSMA-assisted low-altitude MEC systems, where the UAV is despatched to serve $\mathcal{K}\triangleq\{1,2,\ldots,K\}$ GTs. 

\begin{figure}[!t]
\centering
\includegraphics[width=0.38\paperwidth]{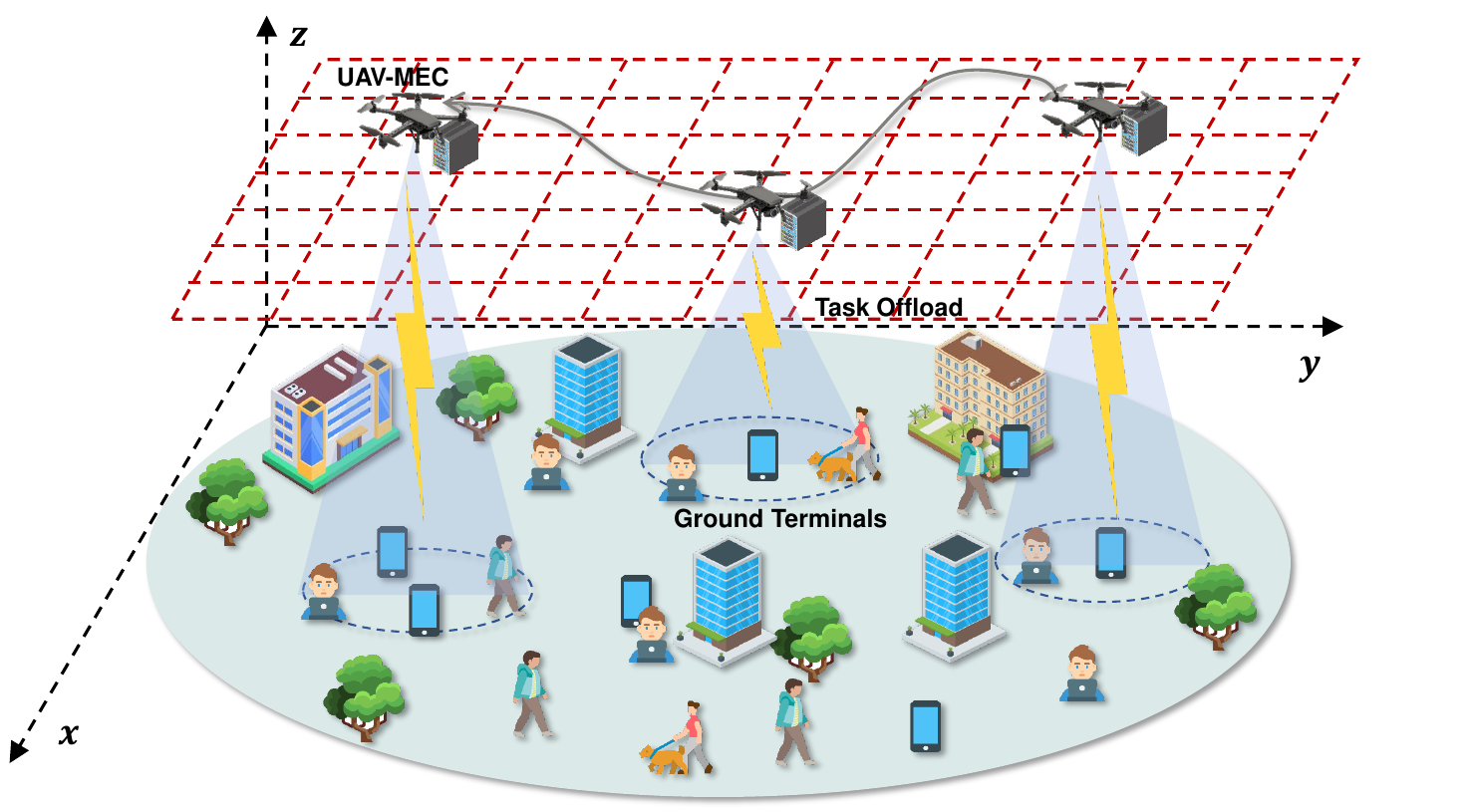}
\caption{RSMA-enabled multi-user low-altitude MEC networks.}
\label{fig:system_model}
\end{figure}

\subsection{Network Model}

The task area is partitioned into a grid of cells with identical area. The horizontal coordinate of the center of the $i$-th cell is represented by $D_{i}^{\ell} = [x_i, y_i]^T \in \mathbb{R}^{2\times1}$, where $\mathcal{L} \triangleq \{1,2,\ldots,L\}$ denotes the set of all cell indices. The spacing between the centers of neighboring cells along the $x$-axis and $y$-axis are denoted by $x_s$ and $y_s$, respectively. The UAV's horizontal position is represented by $D_n^u \in \mathcal{L}$, where $n \in \mathcal{N} \triangleq \{1,2,\ldots,N\}$ and $N$ represents the total number of discrete time slots.
Each time slot is associated with a UAV flight duration level $t_n^u \in \mathcal{T} \triangleq \{1,2,\ldots,T\}$, where the actual duration is bounded by $t_{\min} \leq t_n^u \times \Delta t \leq t_{\max}$. Here, $t_{\min}$ and $t_{\max}$ define the minimum and maximum allowable flight durations, and $\Delta t = \lfloor t_{\max}/T \rfloor$ denotes the time resolution per level. The initial and final UAV positions are fixed and denoted as $D_0^u$ and $D_e^u$, respectively. Accordingly, the horizontal UAV path can be approximated as the sequence $\{D_0^u, D_1^u, \ldots, D_n^u, \ldots, D_N^u, D_e^u\}$.
To capture the UAV's vertical mobility, the altitude at time slot $n$ is represented by $h_n^u \in \mathcal{H} \triangleq \{1,2,\ldots,H\}$. The actual height must satisfy $h_{\min} \leq h_n^u \times \Delta h \leq h_{\max}$, where $h_{\min}$ and $h_{\max}$ denote the minimum and maximum permissible heights, respectively, and $\Delta h = \lfloor h_{\max}/H \rfloor$ denotes the step size between discrete altitude levels. Consequently, the complete UAV trajectory consists of $N$ discrete 3D waypoints $[D_n^u,h_n^u]$ and the associated time durations $t_n^u$, $\forall n \in \mathcal{N}$.

\subsection{Mobility Model}

In UAV-assisted MEC systems, GTs are typically either stationary or exhibit low mobility. Thus, the position of each GT is assumed to be known a priori and represented by $u_{k}=[x_{k},y_{k}]^{T}\in\mathbb{R}^{2\times1},k\in\mathcal{K}$ throughout the entire operation period.
The horizontal velocity of UAV is expressed as
\begin{equation}
v_{n}^{h}=\frac{\|D_{n+1}^{u}-D_{n}^{u}\|}{t_{n}^{u}}\leq V_{\mathrm{max}}^{h}, \forall n\in\mathcal{N},
\end{equation}
where $V_{\mathrm{max}}^{h}$ represents the maximum horizontal velocity. Besides, the vertical flying velocity is expressed as
\begin{equation}
v_{n}^{v}=\frac{\|h_{n+1}^{u}-h_{n}^{u}\|}{t_{n}^{u}}\leq V_{\mathrm{max}}^{v}, \forall n\in\mathcal{N},
\end{equation}
where $V_{\mathrm{max}}^{v}$ denotes the UAV's maximum vertical velocity. When $v_{n}^{h} = 0$, the UAV is either hovering or performing steady straight-and-level flight during time slot $n$.

Based on~\cite{8883173}, the propulsion energy consumption of the UAV is given by
\begin{equation}\label{propulsion_energy}
\begin{aligned}
E^{\text{uav}} = &\sum_{n=1}^{N} t_n^u \Bigg(
    W_0 \left(1 + \frac{3(v_n^h)^2}{U_{\text{t}}^2} \right)
    + \frac{1}{2} F_0 \rho g M (v_n^h)^3 \\
    &+ W_1 \left( 
        \sqrt{1 + \frac{(v_n^h)^4}{4 \bar{v}^4}} - \frac{(v_n^h)^2}{2 \bar{v}^2} \right)^{\frac{1}{2}}+ W_2 v_n^v \Bigg),
\end{aligned}
\end{equation}
where $W_0$ and $W_1$ represent the constant blade profile power and the induced power under hovering conditions, while $W_2$ denotes the constant power required during ascent or descent. The parameters $F_0$ and $g$ represent the fuselage drag ratio and the rotor solidity. In addition, $\rho$ refers to the air density, and $M$ stands for the area of the rotor disc. The parameter $U_{\mathrm{t}}$ denotes the tip velocity of the rotor blade, and $\bar{v}$ signifies the tip velocity of the rotor blade, and $\bar{v}$ refers to the average induced airflow velocity in hovering flight. 

\subsection{Communication Model}

Based on the air-to-ground propagation model for urban scenarios~\cite{6863654}, the probability of establishing the LoS connection between the UAV and GT $k$ can be expressed as
\begin{equation}
\jmath_{k,n}=\frac{1}{1+\alpha\exp\left(-\beta\left(\arctan\left(\frac{h_n^u}{l_{k,n}}\right)-\alpha\right)\right)},
\end{equation}
where $\alpha$ and $\beta$ are environment-dependent constants, $l_{k,n}=\sqrt{(x_i-x_k)^2+(y_i-y_k))^2}$. Based on this, the corresponding pathloss can be expressed as
\begin{equation}
d_{k,n}=20\mathrm{log}\left(\sqrt{(h_n^u)^2+l_{k,n}^2}\right)+A_1\times \jmath_{k,n}+A_2,
\end{equation}
where $A_1=\zeta_{\mathrm{LoS}}-\zeta_{\mathrm{NLoS}}$, $A_2=20\mathrm{log}(\frac{4\pi f_c}{c})+\zeta_{\mathrm{NLoS}}$, $f_c$ denotes the carrier frequency; $c$ represents the speed of light, $\zeta_{\mathrm{LoS}}$ and $\zeta_{\mathrm{NLoS}}$ are the environment-dependent path loss factors for LoS and NLoS conditions, respectively. We define $a_{k,n} \in \{0,1\}$ as the offloading indicator of GT $k$ at time slot $n$, with $a_{k,n}=1$ representing task offloading to the UAV and $a_{k,n}=0$ indicating local computation. To support concurrent offloading, RSMA is applied. Specifically, the message of each GT $w_{k,n}$ is split into $\left |\mathcal{I}  \right |$ sub-messages $w_{k,n,i}$, each allocated a transmit power $p_{k,n,i}$ to achieve rate splitting, such that $\sum_{i \in \mathcal{I}}p_{k,n,i} \leq P_{\mathrm{max}}, \forall k, n$, where $P_{\mathrm{max}}$ is the maximum uplink transmit power of each GT. The composite message is $w_{k,n}=\sum_{i \in \mathcal{I}}\sqrt{p_{k,n,i}}w_{k,n,i},\forall k, n.$ Thus, the received signal of UAV at time slot $n$ is obtained as
\begin{equation}
\begin{aligned}
y_{n}&=\sum_{k\in\mathcal{K}}a_{k,n}\sqrt{g_{k,n}}w_{k,n}+n_{0}\\
&=\sum_{k\in\mathcal{K}}a_{k,n}\sum_{i \in \mathcal{I}}\sqrt{g_{k,n}p_{k,n,i}}w_{k,n,i}+n_0,
\end{aligned}
\end{equation}
where $g_{k,n}$ denotes the channel gain between $k$-th GT and UAV, given by $g_{k,n} = 10^{-d_{k,n}/10}$, where $n_0$ represents the additive white Gaussian noise (AWGN). To decode multiple sub-messages, the UAV employs successive interference cancellation (SIC). Let $\pi_{k,n,i},i\in \mathcal{I}$ indicate the SIC decoding order of sub-message $w_{k,n,i}$, and let $\boldsymbol{\pi}=\{\pi_{k,n,i},\forall k, n, i\}$ denote a decoding permutation, which contains all valid SIC decoding sequences for the $\left |\mathcal{I}  \right |K$ sub-messages. Under this decoding strategy, the achievable uplink rate for sub-message $w_{k,n,i}$ is obtained as~\cite{10742902}
\begin{equation}
R_{k,n,i}=B\log_{2}\left(1+\frac{g_{k,n}p_{k,n,i}}{BN_0+\sum_{(k^{\prime},i^{\prime})\in\mathcal{W}_{k,n,i}}g_{k^{\prime},n}p_{k^{\prime},n,i^{\prime}}}\right),
\end{equation}
where $B$ denotes the system bandwidth, $N_0$ represents the noise power spectral density. The set $\mathcal{W}_{k,n,i}=\{(k^{\prime},i^{\prime})|a_{k^{\prime},n}=a_{k,n}, \pi_{k^{\prime}, n, i^{\prime}} > \pi_{k, n, i}, \forall k^{\prime}, i^{\prime} \}$ contains the sub-messages that are decoded after $w_{k,n,i}$ based on the SIC decoding order at the UAV.

\subsection{Computing Model}
In the considered RSMA-enabled low-altitude MEC systems, each GT has an estimated computational task $G_{k}=\{C_{k},D_{k}\}$, where $C_{k}$ represents the total required CPU cycles, $D_k$ denotes the overall data size to be calculated. As discussed in~\cite{8647789}, when GT $k$ chooses to offload its task to the UAV during time slot $n$, the expected amount of data successfully transmitted in that slot is obtained as
\begin{equation}
m_{k,n}= \sum_{i \in \mathcal{I}} a_{k,n}\kappa_nt_n^uR_{k,n,i},
\end{equation}
where $\kappa_n \in [0,1]$ represents the portion of time slot $n$ that is assigned to data transmission. To ensure that the transmitted data $m_{k,n}$ can be processed by the UAV's onboard CPU with a processing rate of $f_u$, $\kappa_n$ must be properly chosen. This requirement imposes the following constraint:
\begin{equation}
\sum_{i \in \mathcal{I}}\kappa_nt_n^uR_{k,n,i}\leq f_u(1-\kappa_n)t_n^u\frac{D_k}{C_k}.
\end{equation}
To further enhance the total amount of processed data, $\alpha_n$ is further constrained as
\begin{equation}\label{maximization_portion}
\kappa_n=\frac{f_uD_k}{\sum_{i \in \mathcal{I}}R_{k,n,i}C_k+f_uD_k}.
\end{equation}

Given the offloading decision $a_{k,n}$ at time slot $n$ and (\ref{maximization_portion}), the amount of task data $G_k$ processed during slot $n$ is represented as
\begin{equation}\label{task_data}
\varpi_{k,n}=a_{k,n}\frac{f_uD_kt_n^u\sum_{i \in \mathcal{I}}R_{k,n,i}}{\sum_{i \in \mathcal{I}}R_{k,n,i}C_k+f_uD_k}+(1-a_{k,n})t_n^uf_g,
\end{equation}
where $f_g$ denotes the processing speed of the GT's CPU. In (\ref{task_data}), the delay due to downlink transmission from UAV to GT is neglected, assuming that the corresponding data size is negligible. Accordingly, the overall system energy efficiency can be defined as $\eta = \sum_{k=1}^{K}\sum_{n=1}^{N} \varpi_{k,n}/E^{\text{uav}}$.

\subsection{Problem Formulation}

We aim to maximize the system energy efficiency $\eta$ while maintaining reliable MEC performance by jointly optimizing the UAV trajectory, task offloading decisions, uplink power allocation, and SIC decoding order.
Let $\mathbf{A}=\{a_{k,n},\forall k\in\mathcal{K}, \forall n\in\mathcal{N}\}$, $\mathbf{P}=\{p_{k,n,i}, \forall k\in\mathcal{K}, \forall n\in\mathcal{N}, \forall i\}$, $\mathbf{D}=\{D_{n}^{u},\forall n\in\mathcal{N}\}$, $\mathbf{h}=\{h_{n}^{u},\forall n\in\mathcal{N}\}$, and $\mathbf{t}=\{t_{n}^{u},\forall n\in\mathcal{N}\}$, the optimization problem can then be formulated as
\begin{subequations}\label{optimization_problem}
	\begin{align}
  \mathcal{P}0: 
   &\mathop{\mathrm{max}}\limits_{\mathbf{A}, \mathbf{P}, \mathbf{D}, \mathbf{h}, \mathbf{t}, \boldsymbol{\pi}} \ 
		 \eta
		\\
		\mathrm{s.t.} \ 
		\mathrm{C1:} &\ a_{k,n}=\{0,1\},\forall k\in\mathcal{K},\forall n\in\mathcal{N}, \\
        \mathrm{C2:} &\ \sum_{n\in\mathcal{N}}\varpi_{k,n}\geq D_k,\forall k\in\mathcal{K}, \\
        \mathrm{C3:} &\ \sum_{i \in \mathcal{I}}p_{k,n,i} \le P_{\mathrm{max}}, \forall k\in\mathcal{K},\forall n\in\mathcal{N}, \\
        \mathrm{C4:} &\ \frac{\left\|D_{n+1}^u-D_n^u\right\|}{t_n^u}\leq V_{\max}^h,\forall n\in\mathcal{N}, \\ 
        \mathrm{C5:} &\ \frac{\left\|h_{n+1}^u-h_n^u\right\|}{t_n^u}\leq V_{\max}^v,\forall n\in\mathcal{N}, \\
        \mathrm{C6:} &\ h_{\min}\leq h_n^u\times \Delta h\leq h_{\max},\forall n\in\mathcal{N}, \\ 
        \mathrm{C7:} &\ t_{\min}\leq t_n^u\times \Delta t\leq t_{\max},\forall n\in\mathcal{N}, \\
        \mathrm{C8:} &\ p_{k,n,i} \ge 0, \forall k\in\mathcal{K},\forall n\in\mathcal{N}, \forall i\in\mathcal{I}, \\
        \mathrm{C9:} &\ \boldsymbol{\pi} \in \Pi, \\
        \mathrm{C10:} &\ \sum_{i\in\mathcal{I}} R_{k,n,i} \geq R_{\mathrm{min}}, \forall k \in \mathcal{K}, \forall n \in \mathcal{N} \\
        \mathrm{C11:} & \ R_{k,n,1}:R_{k,n,2}:\ldots :R_{k,n,I} = \notag  \\
        &\mu_{k,n,1}:\mu_{k,n,2}:\ldots :\mu_{k,n,I}, \forall k\in\mathcal{K}, \forall n \in \mathcal{N},
        \end{align}
\end{subequations}
where $\mu_{k,n,i}$ denotes a predefined non-negative coefficient that allocates a proportion of the rate $R_{k,n}$ to the $i$-th sub-message of GT $k$ in time slot $n$, satisfying $R_{k,n,i} = \mu_{k,n,i}R_{k,n}, \forall i$ and $\sum_{i \in \mathcal{I}} \mu_{k,n,i} = 1, \forall k, n$. The minimum rate requirement for each GT is denoted by $R_{\mathrm{min}}$. Constraints in $\mathcal{P}0$ are interpreted as follows: C1 models offloading as a binary decision problem; C2 ensures tasks are completed within the UAV's mission time; C3 limits the total transmit power of each GT; C4 $\sim$ C7 describe UAV mobility constraints; C8 enforces non-negative transmission powers; C9 guarantees feasible SIC decoding; C10 enforces the rate requirement; and C11 ensures valid rate splitting among sub-messages.

The optimization problem $\mathcal{P}0$ is a non-convex mixed-integer program due to the discrete offloading variable $\mathbf{a}$, the non-convex objective, and constraint C2, making it difficult to solve optimally. Moreover, the UAV's complex energy model further complicates trajectory and flight time design for energy efficiency maximization. Prior work has shown that solving $\mathcal{P}0$ via successive convex approximation (SCA) incurs high computational cost, rendering it impractical~\cite{8883173}. To overcome these challenges, we derive the optimal decoding order policy, and then develop a GenAI-based optimization framework, detailed in the following section.

\section{The Proposed GDRS Algorithm}\label{algorithm}

To efficiently address problem $\mathcal{P}0$, we decompose it into two sub-problems: 1) the decoding order sub-problem $\mathcal{P}1$, and 2) the joint UAV 3D trajectory, power allocation, and task offloading sub-problem $\mathcal{P}2$. For $\mathcal{P}1$, the optimal decoding policy $\boldsymbol{\pi}$ is determined based on channel gains. Subsequently, $\mathcal{P}2$ is reformulated as a Markov decision process (MDP), and a generative AI-enhanced DRL framework is developed to optimize $\{\mathbf{A}, \mathbf{P}, \mathbf{D}, \mathbf{h}\}$.

\subsection{Decoding Order Policy}

The optimal decoding order plays a critical role in mitigating inter-submessage interference and enhancing overall energy efficiency, and thus must be carefully considered. Given a fixed UAV trajectory $\mathbf{D}$, altitude profile $\mathbf{h}$, offloading decisions $\mathbf{A}$, power allocation $\mathbf{P}$, and UAV flight time $\mathbf{t}$, $\mathcal{P}0$ is rewritten as
\begin{subequations}\label{subproblem_1}
	\begin{align}
  \mathcal{P}1: 
  &\mathop{\mathrm{max}}\limits_{\boldsymbol{\pi}} \ 
		 \eta
		\\
		&\quad \mathrm{s.t.} \ 
		\mathrm{C9}, \mathrm{C10}, \mathrm{C11}.
        \end{align}
\end{subequations}
Solving the decoding order optimization problem $\mathcal{P}1$ is challenging due to its non-convex objective function and the intricate interdependence of discrete variables, rendering it a mixed-integer nonlinear programming (MINLP) problem that is typically hard to solve. While exhaustive search methods~\cite{9257190} can yield the optimal solution, their computational cost scales exponentially with the number of sub-messages, making them impractical for multi-user systems. Existing works~\cite{9240028,9203956,8892657} often adopt channel gain-based heuristics for decoding order. To better balance computational efficiency and system performance, we propose a decoding strategy that jointly considers both the channel gains and the rate-splitting proportions of each GT's sub-messages.
\begin{proposition}\label{proposition_decoding_order}
The decoding order of sub-messages is determined by sorting $\upsilon_{k,n,i} \triangleq \left|{g}_{k,n}\right|^{2}\left(1+\frac{1}{\rho_{k,n,i}^{\min}}\right)\quad=\quad\left|g_{k,n}\right|^{2}\left(1+\frac{1}{2^{\left(\mu_{k,n,i}R_{\min}\right)}-1}\right)$ in descending order, where $\rho_{k,n,i}^{\min}$ represents the minimum signal-to-interference-plus-noise ratio (SINR) required for GT $k$ to decode the $i$-th sub-message at time slot $n$.
\end{proposition}
\begin{proof}
Since the decoding order policy $\boldsymbol{\pi}$ influences the achievable data rate but does not impact the energy consumption of UAV, the optimization problem $\mathcal{P}1$ is equivalently reformulated as
\begin{subequations}\label{subproblem_3_1}
	\begin{align}
  \mathcal{P}1.1: 
  &\mathop{\mathrm{max}}\limits_{\pi_{k,n,i}} \ 
		 \sum_{k \in \mathcal{K}}\varpi_{k,n}
		\\
		&\quad \mathrm{s.t.} \ 
		\mathrm{C9}, \mathrm{C10}, \mathrm{C11}.
        \end{align}
\end{subequations}
As indicated in~(\ref{task_data}), for given UAV trajectory $\mathbf{D}$, altitude profile $\mathbf{h}$, offloading decisions $\mathbf{A}$, power allocation $\mathbf{P}$, and flight time $\mathbf{t}$, the data processed $\varpi_{k,n}$ depends solely on the sum rate $\sum_{i \in \mathcal{I}} R_{k,n,i}$. Hence, maximizing $\varpi_{k,n}$ reduces to maximizing the aggregate rate. To this end, we aim to design a decoding order policy that maximizes the sum rate at each time slot. Let $w_a$ and $w_b$ be two consecutive sub-messages in the decoding sequence, associated with respective channel gains $g_a$ and $g_b$, and corresponding power levels $p_a$ and $p_b$. If the cumulative interference from previously decoded sub-messages is upper bounded by a constant $\epsilon$, we have

\begin{equation}
|g_a|^2 p_a + |g_b|^2 p_b + \xi + N_0 \leq \epsilon,
\end{equation}
where $\xi$ represents the interference from sub-messages decoded after both $w_a$ and $w_b$. Two decoding scenarios are possible: a) $w_a$ is decoded before $w_b$, and b) $w_b$ is decoded before $w_a$. We analyze the achievable signal strength under both cases. Case a): when $w_a$ decoded first:  
In Case a), where $w_a$ is decoded first, the feasible power allocation for $w_a$ and $w_b$ must satisfy
\begin{equation}
\frac{|g_a|^2 p_a}{|g_b|^2 p_b + \xi + N_0} \geq \rho_a^{\min}, \quad
\frac{|g_b|^2 p_b}{\xi + N_0} \geq \rho_b^{\min},
\end{equation}
where $0 \leq p_a, p_b \leq P_{\max}$. Based on this feasible region, the corresponding lower bound of the achievable signal strength can be expressed as
\begin{equation}
p_a \geq \frac{\rho_a^{\min} (\rho_b^{\min} + 1)(\xi + N_0)}{|g_a|^2}, \quad
p_b \geq \frac{\rho_b^{\min} (\xi + N_0)}{|g_b|^2},
\end{equation}
and the associated upper limit for power control is given by
\begin{equation}
p_a \leq P_{\max}, \quad
p_b \leq \min \left\{P_{\max}, \frac{\frac{|g_a|^2 P_{\max}}{\rho_a^{\min}} - \xi - N_0}{|g_b|^2} \right\},
\end{equation}
The maximum achievable signal strength from sub-messages $w_a$ and $w_b$, denoted by $\chi \triangleq |g_a|^2 p_a + |g_b|^2 p_b$, can thus be written as
\begin{equation}\label{C5}
\begin{aligned}
\chi \leq \chi_a \triangleq \min \Big\{ 
&|g_a|^2 P_{\max} \left(1 + \frac{1}{\rho_b^{\min}}\right), \\
&\left(|g_a|^2 + |g_b|^2\right) P_{\max}, \epsilon - \xi - N_0 
\Big\}.
\end{aligned}
\end{equation}

Case b): when $w_b$ is decoded first, the maximum combined signal strength from $w_a$ and $w_b$ can similarly be expressed as
\begin{equation}\label{C6}
\begin{aligned}
\chi \leq \chi_b \triangleq \min \Big\{ 
&|g_b|^2 P_{\max} \left(1 + \frac{1}{\rho_b^{\min}} \right), \\
&\left( |g_a|^2 + |g_b|^2 \right) P_{\max}, \epsilon - \xi - N_0 \Big\}.
\end{aligned}
\end{equation}
By comparing expressions (\ref{C5}) and (\ref{C6}), we observe that in Case a), if $|g_a|^2 \left(1 + \frac{1}{\rho_a^{\min}} \right) \geq |g_b|^2 \left(1 + \frac{1}{\rho_b^{\min}} \right)$, then $\chi_a \geq \chi_b$, i.e., decoding of $w_a$, indicating that decoding $w_a$ first yields higher signal strength. Conversely, in Case b), if $|g_a|^2 \left(1 + \frac{1}{\rho_a^{\min}} \right) \leq |g_b|^2 \left(1 + \frac{1}{\rho_b^{\min}} \right)$, then $\chi_a \leq \chi_b$, and decoding $w_b$ first becomes the better choice.

This confirms that the optimal decoding strategy is to prioritize decoding the sub-message $w_{k,n,i}$ corresponding to the larger value of 
$|g_{k,n}|^2 \left(1 + \frac{1}{\rho_{k,n,i}^{\min}} \right)$. Accordingly, a near-optimal decoding order is obtained by sorting all sub-messages in descending order of $\upsilon_{k,n,i} \triangleq |g_{k,n}|^2 \left(1 + \frac{1}{\rho_{k,n,i}^{\min}} \right) = |g_{k,n}|^2 \left(1 + \frac{1}{2^{\left( \mu_{k,n,i}R_{\min} \right)}-1} \right)$. This concludes the proof.
\end{proof}

\subsection{MDP Construction}

With the decoding order strategy $\boldsymbol{\pi}$ determined, problem $\mathcal{P}0$ is reformulated as
\begin{subequations}\label{subproblem_2}
	\begin{align}
  \mathcal{P}2: 
  &\mathop{\mathrm{max}}\limits_{\mathbf{A}, \mathbf{P}, \mathbf{D}, \mathbf{h}, \mathbf{t}} \ 
		 \eta
		\\
		&\quad \mathrm{s.t.} \ 
		\mathrm{C1} \sim \mathrm{C8}.
        \end{align}
\end{subequations}
To obtain the optimal joint trajectory, task offloading, and power allocation, we propose generative AI-enhanced DRL method.

DRL operates on two core components: the agent and the environment. Their interaction is typically modeled as a MDP, characterized by $\langle\mathcal{S},\mathcal{A},\mathcal{R},\omega,\mathcal{P}\rangle$, where $\mathcal{S}$ denotes the state space, and $\mathcal{A}$ the action space available to the agent. The reward function $\mathcal{R}$ quantifies the immediate feedback received after taking a specific action in a particular state, whereas the discount factor $\omega$ controls the trade-off between short-term and long-term rewards. The transition probability function $\mathcal{P}$ governs the dynamics of state evolution following each action. The key components $\mathcal{S}$, $\mathcal{A}$, and $\mathcal{R}$ are detailed as follows:

\subsubsection{State Spaces}At time slot $n$, the state of GT $k$ is defined by its fixed location, i.e., $s_k(n) = u_{kn} \triangleq u_k, \forall n \in \mathcal{N}$. The state of UAV $s_u(n)$ can be represented by its 3D position and time allocation, i.e., $s_u(n)=(D_n^u,h_n^u, t_n^u)$. Hence, the overall system state at slot $n$ is obtained as $s(n)=\{\{s_{k}(n)\}_{k\in\mathcal{K}},s_{u}(n), n\}\in\mathcal{S}$.

\subsubsection{Action Spaces}Following the mobility model in~\cite{9013924}, the UAV is allowed to move only to one of the adjacent horizontal cells within a single time slot. As a result, its horizontal position $D_u(n)$ is updated in the next slot as:
\begin{equation}
\begin{aligned}
D_u(n+1)=\begin{cases}D_u(n)+(0,y_s),&\mathrm{~if~}d(n)=N\\D_u(n)-(0,y_s),&\mathrm{~if~}d(n)=S\\D_u(n)+(x_s,0),&\mathrm{~if~}d(n)=E\\D_u(n)-(x_s,0),&\mathrm{~if~}d(n)=W\\D_u(n)+(0,0),&\mathrm{~if~}d(n)=I
\end{cases}
\end{aligned}
\end{equation}
where $d(n)$ denotes the UAV's selected horizontal flight action at time slot $n$, with $E$, $S$, $W$, and $N$, indicating movement in the directions of east, south, west, and north, respectively. The symbol $I$ signifies the UAV remains stationary. In the vertical dimension, the UAV is similarly constrained to move only to an adjacent altitude level in each time slot. Accordingly, the UAV's altitude $h_n^u$ will be updated in the next slot as:
\begin{equation}
\begin{aligned}
h_{n+1}^u=\begin{cases}h_n^u+1,&\mathrm{if~}h(n)=U\\h_n^u-1,&\mathrm{if~}h(n)=D\\h_n^u+0,&\mathrm{if~}h(n)=I
\end{cases}
\end{aligned}
\end{equation}
where $h(n)$ denotes the UAV's chosen altitude action at time slot $n$; $U$ and $D$ represent ascending and descending, respectively. If the selected horizontal action $d(n)$ leads the UAV outside the area of interest at time slot $n+1$, it is reset to $I$ to keep the UAV stationary. Similarly, if $h(n)$ causes the UAV's altitude to exceed the allowable range $[h_{\min}, h_{\max}]$, it is set to $I$ to retain the current altitude level. Accordingly, the system action at time slot $n$ can be expressed as $a(n)=(d(n),h(n), t_n^{u}, \{a_{k,n}\}_{k\in\mathcal{K}}, \{p_{k,n,i}\}_{k\in\mathcal{K}}, i \in \{1, 2\})\in \mathcal{A}$, where the action space $\mathcal{A}$ comprises the UAV's movement and time allocation actions, as well as the GTs' task offloading and power control decisions.

\subsubsection{Rewards} The reward associated with a state-action pair $(s(n), a(n))$ at time slot $n$ is given by:
\begin{equation}
\begin{aligned}
r(s(n),a(n))=\lambda_1 \sum_{k\in\mathcal{K}}\frac{\varpi_{k, n+1}} {E_{n+1}^{\mathrm{uav}}} - \lambda_2 PV.
\end{aligned}
\end{equation}
After performing action $a(n)$, the reward $r(s(n), a(n))$ is defined as the energy efficiency, measured by the ratio between the amount of data processed by each GT in the next time slot, i.e., $\varpi_{k,n+1}$, and the UAV’s propulsion energy consumption $E_{n+1}^{\mathrm{uav}}$. To influence the convergence behavior, two scaling factors $\lambda_1$ and $\lambda_2$ are introduced. Additionally, $PV$ denotes a penalty function, which is defined as
\begin{equation}
\begin{aligned}
\left.PV=\left\{\begin{array}{ll}c_0,&\mathrm{if \ C1 \sim C8 \ are \ not \ satisfied},\\0,&\mathrm{Otherwise},\end{array}\right.\right.
\end{aligned}
\end{equation}
where $c_0$ is a positive constant that controls the penalty magnitude. Through this design, the DRL framework aims to learn energy-efficient solutions for RSMA-enabled low-altitude MEC systems.

\subsection{The Proposed Generative AI-enhanced DRL Method}

Jointly optimizing UAV trajectory, task offloading, and power allocation presents significant challenges for conventional DRL algorithms, primarily due to the high-dimensional and discrete nature of the action space. Standard DRL agents often exhibit poor sample efficiency and limited exploration capability under such structured constraints. To address these issues, we adopt a diffusion model-based DRL framework that utilizes generative modeling via a learned denoising process to generate structured actions. This method promotes more efficient exploration and supports the synthesis of valid, high-quality policies that better satisfy system constraints and optimize long-term energy efficiency.

The Denoising Diffusion Probabilistic Model (DDPM) progressively perturbs data into Gaussian noise through a forward diffusion process and then learns to recover the original data via a reverse denoising process~\cite{nichol2021improved}. Leveraging this generative mechanism and its ability to incorporate conditional information, we design a diffusion-based optimizer for enhancing solution quality and synergizing effectively with DRL to address dynamic and complex optimization in RSMA-enabled low altitude MEC systems.

In the diffusion framework, an optimal decision under the current environment is progressively perturbed with noise until it becomes Gaussian, a process referred to as forward probability noising~\cite{10529221}. In the reverse phase, the decision generation network $\pi_{\theta}(\cdot)$ serves as a denoiser that reconstructs the original solution $\boldsymbol{z}_0$ from Gaussian noise, conditioned on the system state $s$. The following presents the formal mathematical formulation of this diffusion-based decision process.

\subsubsection{Forward Process} The decision output $\boldsymbol{z}_0=\pi_{\vartheta}(s)$ represents the probability distribution over discrete decisions under the observed environment state $s$.  In the forward diffusion process, this initial distribution is gradually perturbed by Gaussian noise, resulting in a sequence of intermediate representations $\boldsymbol{z}_1,\boldsymbol{z}_2,\ldots,\boldsymbol{z}_T$, each sharing the same dimensionality as $\boldsymbol{z}_0$. At each step $t$, the transition from $\boldsymbol{z}_{t-1}$ to $\boldsymbol{z}_t$ follows a Gaussian distribution with mean $\sqrt{1-\varphi_{t}}\boldsymbol{z}_{t-1}$ and variance $\varphi_t\mathbf{I}$, as described in~\cite{ho2020denoising}.
\begin{equation}
q\left(\boldsymbol{z}_{t}|\boldsymbol{z}_{t-1}\right)=\mathcal{N}\left(\boldsymbol{z}_{t};\sqrt{1-\varphi_{t}}\boldsymbol{z}_{t-1},\varphi_{t}\mathbf{I}\right),
\end{equation}
where $t=1,\ldots,T$, $\varphi_{t}=1-e^{-\frac{\varphi_{\min}}{T}-\frac{2t-1}{2T^{2}}(\varphi_{\max}-\varphi_{\min})}$ denotes the time-dependent variance in the forward process~\cite{ho2020denoising}.

Since each $\boldsymbol{z}_t$ depends only on its immediate predecessor $\boldsymbol{z}_{t-1}$, the forward process constitutes a Markov chain. Consequently, the distribution of $\boldsymbol{z}_T$ conditioned on $\boldsymbol{z}_0$ can be expressed as the product of transition probabilities across denoising steps~\cite{ho2020denoising}, which is shown as 
\begin{equation}
q(\boldsymbol{z}_T|\boldsymbol{z}_0) = \prod_{t=1}^T q(\boldsymbol{z}_t|\boldsymbol{z}_{t-1}).
\end{equation}

Although the forward process is not explicitly executed, it defines a closed-form relationship between $\boldsymbol{z}_0$ and any intermediate state $\boldsymbol{z}_t$, given by
\begin{equation}\label{forward_process}
\boldsymbol{z}_t=\sqrt{\bar{\nu}_t}\boldsymbol{z}_0+\sqrt{1-\bar{\nu}_t}\boldsymbol{\epsilon},
\end{equation}
where $\nu_{t}=1-\varphi_{t}$, $\bar{\nu}_{t}=\prod_{k=1}^{t}\nu_{k}$ denotes the cumulative product up to step $t$. The forward process relates $\boldsymbol{z}_t$ and $\boldsymbol{z}_0$ as a noisy transformation, where $\varepsilon\sim\mathcal{N}(\mathbf{0},\mathbf{I})$ is standard Gaussian noise. As $t$ increases, $\boldsymbol{z}_T$ converges to pure noise distributed as $\mathcal{N}(\mathbf{0}, \mathbf{I})$. However, since wireless network optimization problems typically lack ground-truth datasets of optimal decisions $\boldsymbol{z}_0$, the forward process is not executed in the proposed framework.

\subsubsection{Reverse Process} The goal of the reverse process is to reconstruct the desired target $\boldsymbol{z}_0$ starting from a noise sample $\boldsymbol{z}_T \sim \mathcal{N}(\mathbf{0}, \mathbf{I})$ by iteratively denoising it. Within our framework, this procedure corresponds to synthesizing the optimal decision policy from an initially Gaussian-distributed sample. The transition between $\boldsymbol{z}_t$ and $\boldsymbol{z}_{t-1}$ is modeled as $p(\boldsymbol{z}_{t-1}|\boldsymbol{z}_{t})$, which is intractable in closed form but can be approximated by a Gaussian distribution, expressed as
\begin{equation}\label{reverse_process_Gaussian}
p_{\vartheta}\left(\boldsymbol{z}_{t-1}|\boldsymbol{z}_{t}\right)=\mathcal{N}\left(\boldsymbol{z}_{t-1};\boldsymbol{\mu}_{\vartheta}\left(\boldsymbol{z}_{t},t,s\right),\tilde{\varphi}_{t}\mathbf{I}\right),
\end{equation}
where the mean $\boldsymbol{\mu}_{\vartheta}$ is learned via a deep neural network, and the variance is obtained as~\cite{ho2020denoising}.
\begin{equation}\label{variance_of_reverse_process}
\tilde{\varphi}_t=\frac{1-\bar{\nu}_{t-1}}{1-\bar{\nu}_{t}}\varphi_{t}
\end{equation}.

By applying Bayes' theorem, the reverse process can be reformulated in terms of the forward process and expressed as a Gaussian probability density function. Accordingly, the mean is derived as
\begin{equation}\label{transition_Gaussian}
\boldsymbol{\mu}_{\vartheta}\left(\boldsymbol{z}_t,t,s\right)=\frac{\sqrt{\nu_t}\left(1-\bar{\nu}_{t-1}\right)}{1-\bar{\nu}_t}\boldsymbol{z}_t+\frac{\sqrt{\bar{\nu}_{t-1}}\varphi_t}{1-\bar{\nu}_t}\boldsymbol{z}_0,
\end{equation}
where $t=1,\ldots,T$. Based on the forward process in~(\ref{forward_process}), the reconstructed sample $\boldsymbol{z}_0$ can be directly estimated by
\begin{equation}\label{reconstructed_sample}
\boldsymbol{z}_0=\frac{1}{\sqrt{\bar{\nu}_t}}\boldsymbol{z}_t-\sqrt{\frac{1}{\bar{\nu}_t}-1}\cdot\tanh\left(\boldsymbol{\sigma}_{\vartheta}(\boldsymbol{z}_t,t,s)\right),
\end{equation}
where $\boldsymbol{\sigma}_{\vartheta}(\boldsymbol{z}_t,t,s)$ denotes a deep neural network parameterized by $\vartheta$ that predicts the denoising noise conditioned on the observed state $s$. To prevent excessive noise in the reconstructed decision $\boldsymbol{z}_0$, which may obscure the true action probabilities, the output is scaled using a hyperbolic tangent activation to ensure bounded noise levels.

During the reverse denoising process, each step $t$ introduces a new noise component $\boldsymbol{\sigma}_{\vartheta}$, which remains independent of the forward-process noise $\boldsymbol{\sigma}$. As a result, $\boldsymbol{z}_0$ cannot be directly recovered using~(\ref{reconstructed_sample}). Instead, we substitute~(\ref{reconstructed_sample}) into the reverse transition expression in~(\ref{transition_Gaussian}) to estimate the mean, shown as
\begin{equation}\label{mean_of_reverse_process}
\boldsymbol{\mu}_{\vartheta}\left(\boldsymbol{z}_{t},t,s\right)=\frac{1}{\sqrt{\nu_{t}}}\left(\boldsymbol{z}_{t}-\frac{\varphi_{t}\tanh\left(\boldsymbol{\sigma}_{\vartheta}(\boldsymbol{z}_{t},t,s)\right)}{\sqrt{1-\bar{\nu}_{t}}}\right).
\end{equation}

We then sample $\boldsymbol{z}_{t-1}$ from the reverse transition distribution $p(\boldsymbol{z}_t)p_{\vartheta}(\boldsymbol{z}_{t-1}|\boldsymbol{z}_t)$ and iterate this process over $t = T, T-1, \ldots, 1$. By recursively applying these steps, the final generation distribution $p_{\vartheta}(\boldsymbol{z}_0)$ is given by
\begin{equation}
p_{\theta}\left(\boldsymbol{z}_0\right)=p\left(\boldsymbol{z}_T\right)\prod_{t=1}^Tp_{\theta}\left(\boldsymbol{z}_{t-1}|\boldsymbol{z}_t\right),
\end{equation}
where $p\left(\boldsymbol{z}_T\right)$ denotes a standard Gaussian distribution. After the reverse process yields the generative distribution $p_{\vartheta}(\boldsymbol{z}_0)$, a sample of the final output $\boldsymbol{z}_0$ can be drawn accordingly.

A common challenge in training generative models lies in the inability to backpropagate gradients through stochastic sampling operations. To overcome this, we adopt a reparameterization technique that separates the source of randomness from the distribution parameters. Specifically, the sampling is reformulated using the following update rule:
\begin{equation}\label{update_rule}
\boldsymbol{z}_{t-1}=\boldsymbol{\mu}_{\vartheta}\left(\boldsymbol{z}_t,t,s\right)+\left(\tilde{\varphi}_t/2\right)^2\odot\boldsymbol{\sigma},
\end{equation}
where $\sigma\sim\mathcal{N}(0,\mathbf{I})$. By recursively adopting the reverse update rule in~(\ref{update_rule}), all intermediate states $\boldsymbol{z}_{t} \ (0 \le t \le T)$, including the final output $\boldsymbol{z}_0$, can be generated from an initial Gaussian noise sample.

Finally, the softmax function is applied to $\boldsymbol{z}_0$ to convert it into a valid probability distribution, given by
\begin{equation}\label{probability_distribution_softmax}
\pi_{\vartheta}(s)=\left\{\frac{e^{\boldsymbol{z}_0^j}}{\sum_{i=1}^{\mathcal{A}}e^{\boldsymbol{z}_0^i}},\forall j\in\mathcal{A}\right\}
\end{equation}
Each element in $\pi_{\vartheta}(s)$ represents the probability of selecting a specific action under state $s$.


To implement the proposed approach in practice, the first step involves calculating the mean $\boldsymbol{\mu}_{\vartheta}$ of the reverse transition distribution $p_{\vartheta}(\boldsymbol{z}_{t-1}|\boldsymbol{z}_t)$, as defined in (\ref{reverse_process_Gaussian}) and (\ref{mean_of_reverse_process}), and then update $\boldsymbol{z}_{t-1}$ using the rule in~(\ref{update_rule}). Subsequently, the optimal decision distribution $\boldsymbol{z}_{0}$ is obtained via~(\ref{probability_distribution_softmax}). To further enhance optimization, we integrate the diffusion model into the Soft Actor-Critic (SAC) framework. SAC is an off-policy DRL algorithm that augments the reward with an entropy term to jointly maximize expected returns and policy entropy, effectively balancing exploration and exploitation.



\subsubsection{Experience Replay} The agent continuously interacts with the environment to collect trajectory data. At slot $n$, the actor observes the state $s(n)$ and generates a discrete action distribution $\pi_{\vartheta}(s(n))$. An action $a(n) \sim \pi_{\vartheta}(s(n))$ is sampled and executed, after which the environment returns a reward $r(n) = (s(n), a(n))$ and transitions to the next state $s(n+1)$. The resulting transition tuple $(s(n),a(n),r(n),s(n+1)$ is stored in the experience replay buffer $\mathcal{B}$. This replay mechanism enables real-time interaction while allowing asynchronous sampling of past experiences, thereby enhancing training efficiency.

\subsubsection{Double Critic Networks} In the proposed framework, the critic network is implemented as a soft Q-function $Q_\psi(s(n), a(n))$, which evaluates the expected return augmented by the entropy of the policy $\psi$ under $(s(n), a(n))$. The formulation is given by
\begin{equation}
Q_{\psi}(s(n),a(n))=E_{s(n+1)\sim\mathcal{B}}[r(s(n),a(n))+\omega V_{\pi}(s(n+1)))],
\end{equation}
where $\omega$ denotes the reward discount factor, and $s(n+1)$ is the subsequent state sampled from the replay buffer $\mathcal{B}$. The corresponding soft value function $V_{\pi}(s(n))$ is expressed as:
\begin{equation}
\begin{aligned}&V_{\pi}(s(n))=E_{a(n)\sim\pi}[Q_{\psi}(s(n+1),a(n+1))+\gamma J(\pi_{\vartheta}(s(n)))],\\&\mathrm{s.t.}J(\pi_{\vartheta}(s(n)))=-\pi_{\vartheta}(s(n))\log\pi_{\vartheta}(s(n)),\end{aligned}
\end{equation}
the term $J(\pi_{\vartheta}(s(n)))$ represents the entropy of the policy at state $s(n)$, while $\gamma \in [0,1]$ is a tunable temperature coefficient that regulates the influence of the entropy in the objective. To reduce the overestimation bias often observed in Q-learning algorithms, GDRS incorporates a pair of critic networks, $Q_{{\psi}_1}(s(n), a(n))$ and $Q_{{\psi}_2}(s(n), a(n))$. The actor is trained using the smaller of the two Q-values, ensuring a conservative and stable policy update.

\subsubsection{Target Networks} Each critic network is associated with a corresponding target network, parameterized by $\hat{\psi}_{1}$ and $\hat{\psi}_{2}$, respectively. These target networks mirror the architecture of their online counterparts and are introduced to enhance training stability by reducing fluctuations in target estimates. To ensure smooth updates, their parameters are adjusted incrementally using a soft update strategy, defined as:
\begin{equation}\label{update_rule_of_target}
\hat{\psi_i}\leftarrow\varsigma \psi_i+(1-\varsigma )\hat{\psi_i},
\end{equation}
where $\varsigma  \in (0,1]$ denotes the soft update rate, and $i=1,2$.

To train the critic networks, the agent draws a mini-batch of transition samples from the replay buffer $\mathcal{B}$ and optimizes the network parameters by minimizing the following loss function:
\begin{equation}\label{update_rule_of_critic}
L_{Q}(\psi_{i})=E_{(s(n),a(n))\sim\mathcal{B}}[\frac{1}{2}(Q_{\psi_{i}}(s(n),a(n))-\hat{q})^{2}],
\end{equation}
where we define the Q-target $\hat{y}$ as
\begin{equation}
\hat{q}=r(s(n),a(n))+\omega V_{\hat{\psi}}(s(n+1)).
\end{equation}
Subsequently, the parameters of the two critic networks are updated using gradient descent, following the update rule:
\begin{equation}\label{update_rule_of_double_critic}
\psi_i\leftarrow\psi_i-\tau_c\nabla_{\psi_i}L_Q(\psi_i),
\end{equation}
where $\tau_c$ represents the learning rate of the critic network for updating critic network parameters.

For training the policy network, the actor aims to maximize the expected Q-value to enhance decision-making quality. The objective function for updating the actor is defined as:
\begin{equation}\label{loss_function_of_actor}
L_{\pi}(\vartheta)=E_{s(n)\sim\mathcal{B}}[\gamma J(\pi_{\vartheta}(s(n)))-\min_{i=1,2}Q_{\psi_{i}}(s(n),a(n))],
\end{equation}
and we update the parameter $\vartheta$, shown as
\begin{equation}\label{update_rule_of_actor}
\vartheta\leftarrow\vartheta-\tau_a\nabla_\vartheta L_\pi(\vartheta),
\end{equation}
where $\tau_a$ indicates the learning rate of actor network.

\begin{figure}[!t]
\centering
\includegraphics[width=0.41\paperwidth]{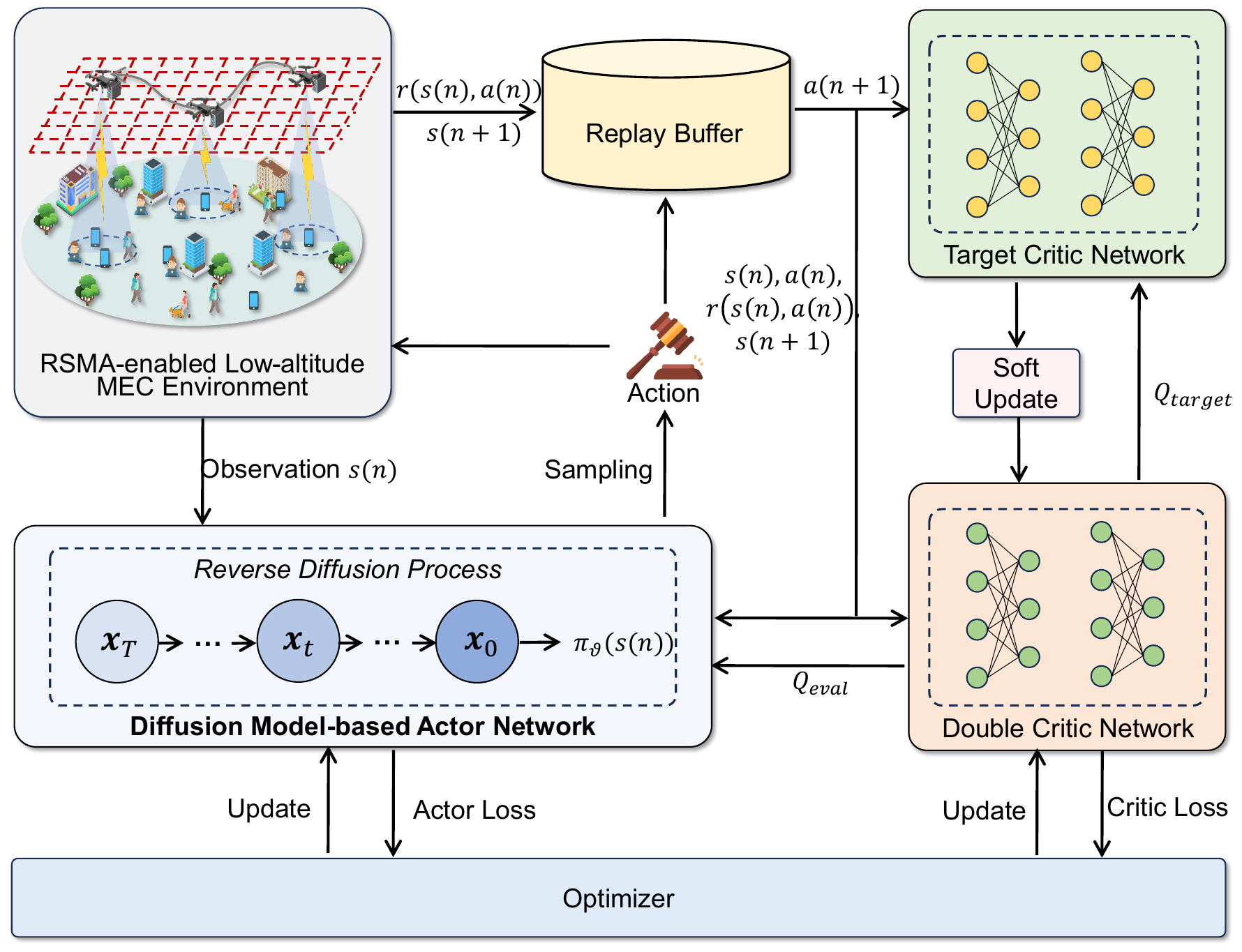}
\caption{The proposed GDRS optimization framework for joint optimization of UAV trajectory, task offloading and power allocation in RSMA-enabled low-altitude MEC systems.}
\label{proposed_framework}
\end{figure} 

\subsection{Algorithm Framework}

Fig.~\ref{proposed_framework} shows the framework of the proposed GDRS method, which extends the SAC algorithm. It comprises several core components for policy training, including an diffusion model-based actor network, double critic networks with corresponding target networks, an experience replay buffer, and the interaction environment~\cite{10409284}. In addition, the detailed progress of the proposed GDRS method is presented in Algorithm~\ref{training_algorithm}.

\begin{algorithm}[!t]
    \caption{GDRS: Generative Diffusion-Enhanced DRL for RSMA-enabled Low-Altitude MEC}
    \label{training_algorithm}
    \SetKwInput{KwInput}{Input}
    \SetKwInOut{KwOutput}{Output}

    \KwInput{Initialize replay buffer $\mathcal{B}$; actor parameters $\vartheta$; critic parameters $\phi_1, \phi_2$; target parameters $\hat{\psi}_1 \leftarrow \psi_1$, $\hat{\psi}_2 \leftarrow \psi_2$; training step index $e$; diffusion step index $t$; trajectory counter $q$.}

    \While{$e < E$}{
        \While{$q < Q$}{
            Determine the optimal decoding order $\boldsymbol{\pi}$ according to Proposition~\ref{proposition_decoding_order}; \\
            Observe current state $s(n)$; sample initial noise $\boldsymbol{z}_T \sim \mathcal{N}(\mathbf{0}, \mathbf{I})$; \\

            \While{$t < T$}{
                Apply denoising model $\boldsymbol{\sigma}_\vartheta(\boldsymbol{z}_t, t, s(n))$; \\
                Compute mean $\boldsymbol{\mu}_\vartheta$ and variance $\tilde{\varphi}_t$ using (\ref{mean_of_reverse_process}) and (\ref{variance_of_reverse_process}); \\
                Update $\boldsymbol{z}_t$ using (\ref{update_rule}); \\
                $t \gets t + 1$;
            }

            Sample discrete action $a(n)$ to determine UAV trajectory, offloading, and power allocation; \\
            Execute $a(n)$ in the environment; record the resulting reward $r(n)$ and the subsequent state $s'(n)$; \\
            Store the transition tuple $(s(n), a(n), r(n), s'(n))$ in $\mathcal{B}$; \\

            Randomly select a mini-batch of samples from $\mathcal{B}$; \\
            Update critic parameters $\psi_1$, $\psi_2$ using (\ref{update_rule_of_critic}), (\ref{update_rule_of_double_critic}); \\
            Update actor parameters $\vartheta$ using (\ref{loss_function_of_actor}) and (\ref{update_rule_of_actor}); \\
            Perform soft updates on $\hat{\psi}_1$, $\hat{\psi}_2$ via (\ref{update_rule_of_target}); \\

            $q \gets q + 1$;
        }
        $e \gets e + 1$;
    }
    \KwOutput{Optimized actor policy $\vartheta^*$.}
\end{algorithm}

\section{Simulation Results and Discussion}\label{sec_simulation}

This section presents simulation studies designed to evaluate the effectiveness of the proposed GDRS approach. We begin by outlining the experimental configuration and baseline settings. Then, the simulation outcomes are discussed and examined thoroughly.

\subsection{Experimental Setup} 
\subsubsection{Experimental Configuration} We consider a task area measuring $1000 \times 1000~\mathrm{m}$, which includes one UAV equipped with an MEC server and certain number of GTs. We set the initial location of the UAV as $\left[0,0,200\right]$. The total time slots of the system is set to $100$, and each time slot is of length $1~\mathrm{s}$. According to the settings about propulsion model of the rotary-wing UAV~\cite{8883173}, we set $U_{\mathrm{t}} = 120~\mathrm{m/s}$, $\bar{v} = 4.03~\mathrm{m/s}$, $F_0 = 0.6$, $g=0.05$, $\rho=1.225$, $M=0.503$. Each GT divides its transmitted signal into two distinct sub-messages, i.e., $I = 2$. The rate-splitting ratio of the sub-messages is set to $\mu_{k,n,i}=0.5, \forall k,n,i$. As for the training model, we set the learning rates of actor network and critic networks as $\tau_c = \tau_a = 5\times10^{-4}$, and set the update rate $\varsigma=5\times10^{-3}$ for soft updating the target networks.
All experiments were conducted in a Python 3.10 environment with PyTorch 2.1.0 on a server running Ubuntu, equipped with a 2.10 GHz Intel Xeon Gold 5218R processor (40 cores, 503 GB RAM) and an NVIDIA A100 GPU featuring 80 GB of memory. More parameters about the considered scenario and proposed method are shown in Table~\ref{table_parameters}.

\begin{table}[!t]\footnotesize
	\renewcommand{\arraystretch}{1.3}
	\caption{System Parameters}
	\label{table_parameters}
	\centering
	\begin{tabular}{p{0.6\linewidth}|p{0.3\linewidth}}
		\hline\hline
		\textbf{Parameters}
		& \textbf{Values}
		\\ \hline
            Carrier frequency $f_c$
		& $2.4~\mathrm{GHz}$
            \\ \hline
            Environment constants $\alpha$, $\beta$
		& $12.08$, $0.11$~\cite{6863654}
            \\ \hline
            LoS and NLoS pathloss factors $\zeta_{\mathrm{LoS}}$, $\zeta_{\mathrm{NLoS}}$ & $1.6$, $23$~\cite{10778659}
            \\ \hline
            Communication bandwidth $B$
		& $1~\mathrm{MHz}$
		\\ \hline
            Maximum transmit power of each GT $P_{\mathrm{max}}$ & $5~\mathrm{mW}$
            \\ \hline
            Speed of light $c$
		& $3 \times 10^{8}~\mathrm{m/s}$
            \\ \hline
            Noise power $N_0$
		& $-174~\mathrm{dBm/Hz}$
            \\ \hline
            UAV flight time $t_{\mathrm{min}}$, $t_{\mathrm{max}}$
		& $1~\mathrm{s}$, $5~\mathrm{s}$
            \\ \hline
            UAV flight height $h_{\mathrm{min}}$, $h_{\mathrm{max}}$
		& $100~\mathrm{m}$, $200~\mathrm{m}$
            \\ \hline
            Maximum flight speed $V_{\mathrm{max}}^{h}$, $V_{\mathrm{max}}^{v}$
		& $10~\mathrm{m/s}$, $10~\mathrm{m/s}$
            \\ \hline
            Blade profile power $W_0$
		& $79.9~\mathrm{W}$~\cite{8883173}
            \\ \hline
            Induced power $W_1$
		& $88.6~\mathrm{W}$~\cite{8883173}
            \\ \hline
            Descending/ascending power $W_2$
		& $11.46~\mathrm{W}$~\cite{8883173}
            \\ \hline
            Processed rate of GT $f_g$
		& $5~\mathrm{cycles/s}$
            \\ \hline
            Processed rate of UAV $f_u$
		& $100\sim500~\mathrm{cycles/s}$
            \\ \hline
            The CPU cycles to be computed $C_k$ & $500 \sim 2500~\mathrm{bits}$
            \\ \hline
            The amount of data to be processed $D_k$ & $1000 \sim 1500~\mathrm{cycles}$
            \\ \hline
            Number of GTs $K$ & $2 \sim 5$
            \\ \hline
            Total time slots $N$ & $100$
		\\ \hline
            Minimal rate requirement $R_{\mathrm{min}}$ & $0.2~\mathrm{bit/s/Hz}$
            \\ \hline
		  The number of denoising steps $T$
		& $10$, $20$, $30$
            \\ \hline
		  Learning rate $\tau_{c}$, $\tau_{a}$
		& $5\times10^{-4}$
            \\ \hline
		  Weight of soft update $\varsigma$
		& $0.005$
            \\ \hline
		  Discount factor $\omega$
		& $0.95$
            \\ \hline
		  Number of training episodes $E$ & $500$
            \\ \hline
		  Batch size & $64$
		\\ \hline\hline
	\end{tabular}
\end{table}

\subsubsection{Baseline Settings} To present fair comparison, we compare our proposed GDRS method with the following well-known DRL benchmark:
\begin{itemize}
\item PPORS: Proximal Policy Optimization for RSMA-enabled low-altitude MEC, where the agent learns an optimal policy in a discrete action space to maximize the energy efficiency.
\item DQNRS: Deep Q-Network for RSMA-enabled low-altitude MEC, where the agent approximates the optimal action-value function in a discrete action space to maximize the energy efficiency.
\item GDRS-Random: Generative diffusion-enhanced optimization for RSMA-enabled low-altitude MEC, where UAV adopts the random decoding order for uplink GTs instead of optimal policy proposed in \textbf{Proposition~\ref{proposition_decoding_order}}.
\end{itemize}

Additionally, in order to validate the effectiveness of the RSMA framework in supporting low-altitude MEC systems, we also design generative diffusion-enhanced algorithms under different multiple access schemes as baseline approaches:
\begin{itemize}
\item GDFD: Generative diffusion-enhanced optimization for Frequency Division Multiple Access (FDMA)-enabled low-altitude MEC, where the uplink bandwidth is equally divided among GTs, ensuring orthogonal transmissions without inter-user interference.
\item GDNO: Generative diffusion-enhanced optimization for NOMA-enabled low-altitude MEC, where multiple GTs simultaneously share the same uplink resources via power-domain superposition coding and successive interference cancellation at the UAV.
\end{itemize}

\subsection{Experimental Results}

Fig.~\ref{fig:convergence} plots the average reward obtained by GDRS and the baseline algorithms over training episodes. First, GDRS converges substantially faster and attains a higher steady-state reward than all baselines. Specifically, the reward curve of GDRS rises sharply and stabilizes at a high value after about 110 training episodes. Furthermore, the proposed GDRS method outperforms PPORS and DQNRS due to its enhanced exploration capability in high-dimensional hybrid action spaces, achieved by introducing a diffusion model to generate structured and high-quality action candidates. In addition, compared to the method with random decoding order, GDRS exhibits less fluctuation in its training curve and achieves a higher average reward, demonstrating the effectiveness of the optimized decoding order.

\begin{figure}[!t]
\centering
\includegraphics[width=0.95\linewidth]{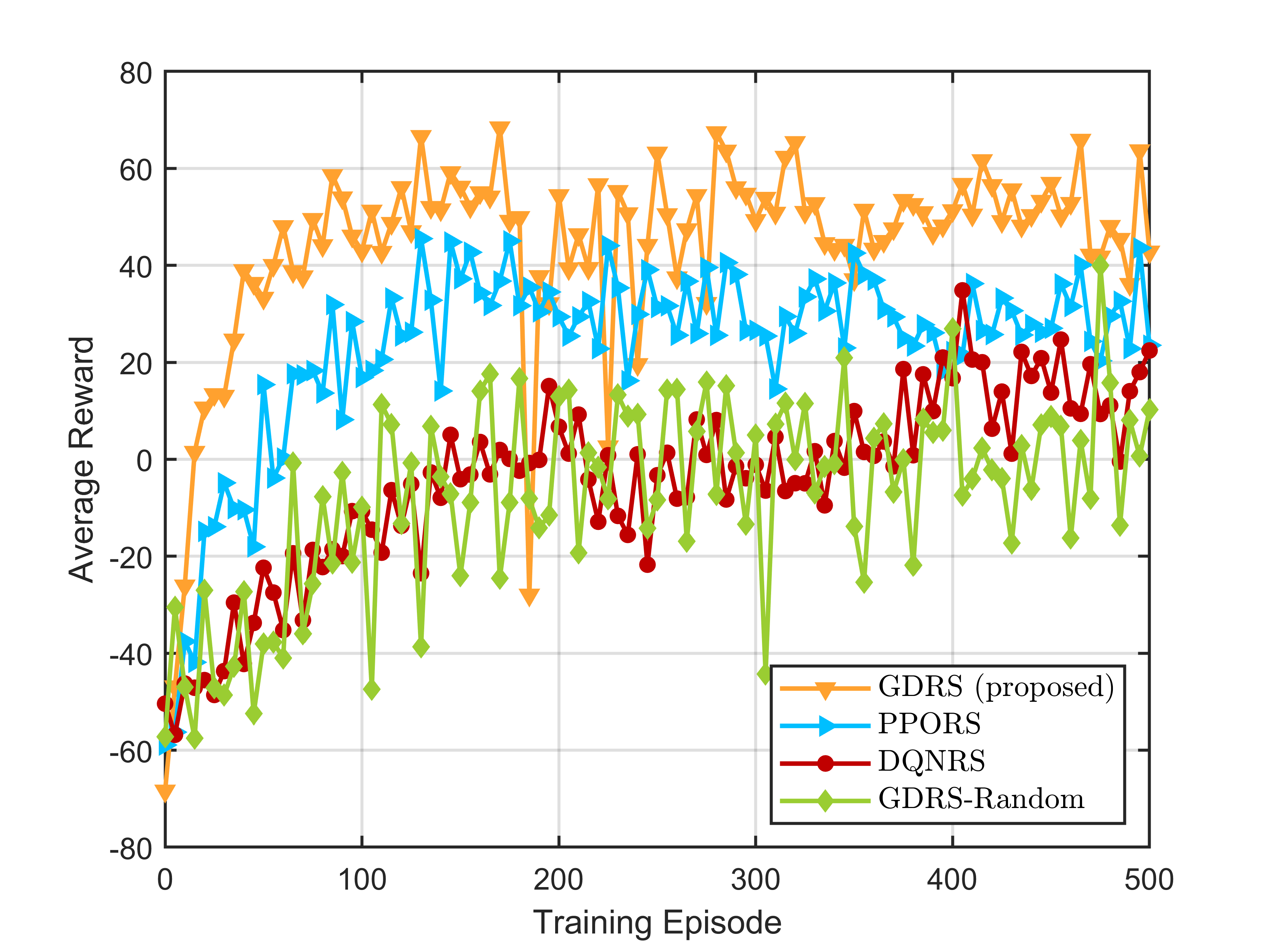}
\caption{The reward curves of the proposed GDRS algorithm and other baselines for RSMA-enabled LAE MEC systems over training episodes.}
\label{fig:convergence}
\end{figure}

Fig.~\ref{convergence_different_learning_rate} examines the training reward curves of the proposed GDRS algorithm under different learning rates. When the learning rate is set to $5\times 10^{-4}$, the training curve achieves higher reward values compared to those with learning rates of $1\times 10^{-4}$ and $1\times 10^{-3}$. This is because a large learning rate causes excessively large update steps, resulting in unstable training for the agent. Conversely, a small learning rate leads to a slow learning process and causes the training to converge prematurely or get trapped in a local optimum. A learning rate of $5\times10^{-4}$ strikes a better balance between learning stability and update efficiency. Therefore, all subsequent experiments adopt a learning rate of $5\times 10^{-4}$.

\begin{figure}[!t]
\centering
\includegraphics[width=0.95\linewidth]{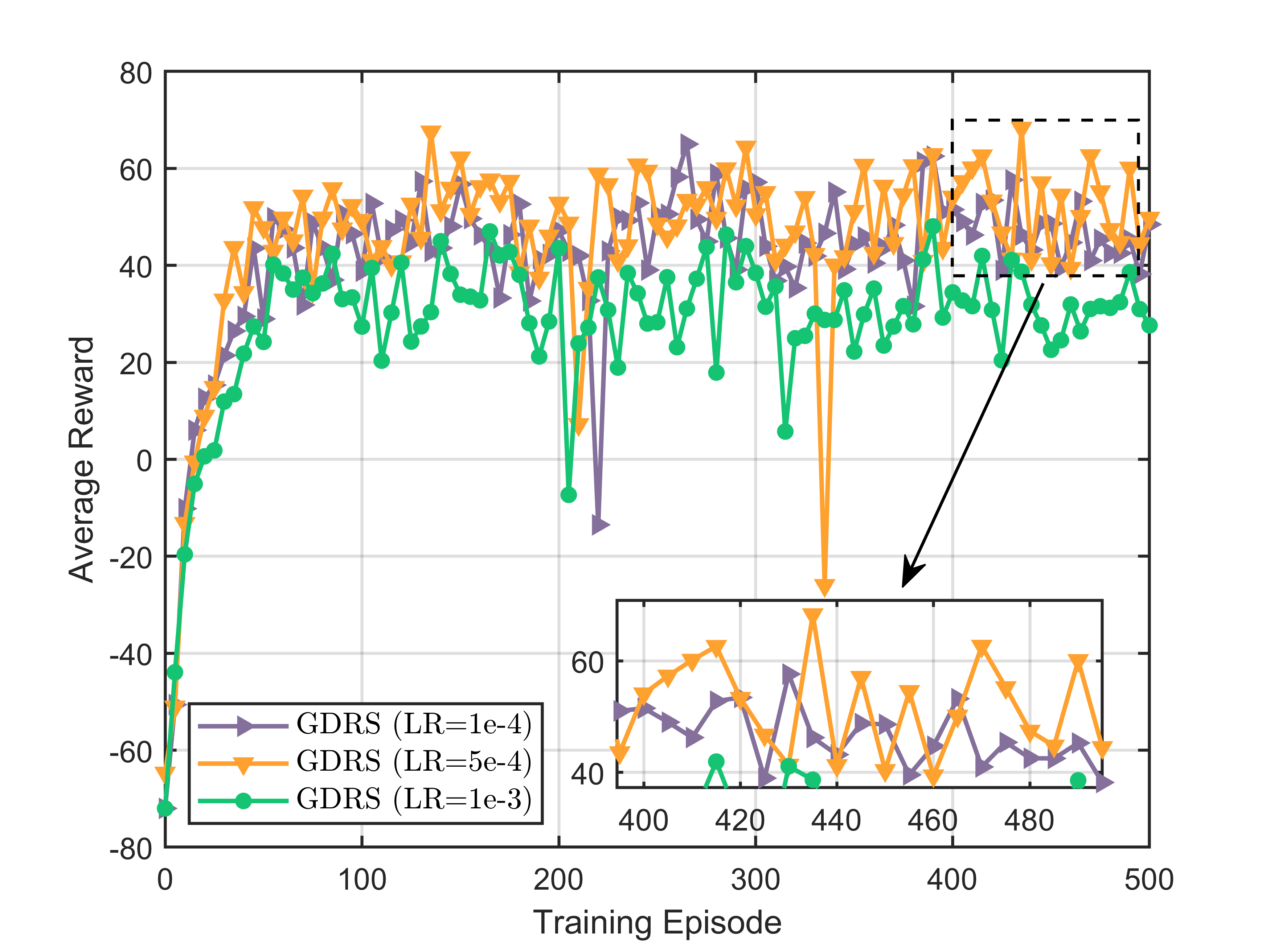}
\caption{The reward curves of the proposed GDRS algorithm under different learning rate.}
\label{convergence_different_learning_rate}
\end{figure}

\begin{figure}[!t]
\centering
\includegraphics[width=0.95\linewidth]{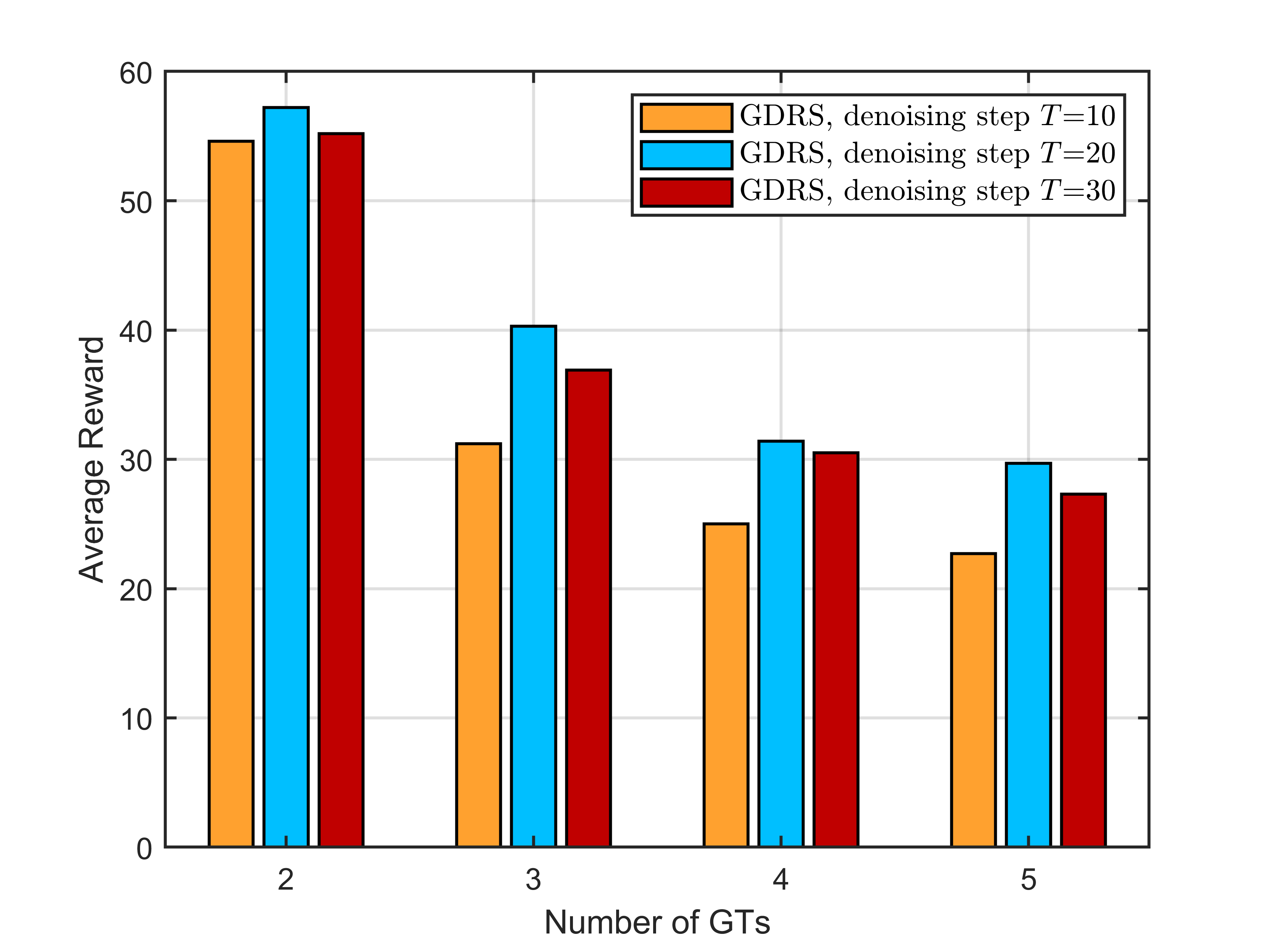}
\caption{The impact of diffusion steps $T$ on average reward under different number of GTs $K$.}
\label{reward_vs_diffusion_steps}
\end{figure}

Fig.~\ref{reward_vs_diffusion_steps} investigates how the number of diffusion steps $T$ affects the achieved average reward under different numbers of GTs $K$. First, we observe that the GDRS method with 20 denoising steps outperforms the version with 10 denoising steps, indicating that increasing the number of denoising steps helps the DRL agent learn more generalizable features, thereby enhancing training generalization. However, the GDRS method with 30 denoising steps yields a significantly lower average reward compared to the 20-step case. This is because an excessive number of denoising steps can remove valuable features from the data, reducing training efficiency. Furthermore, as the number of GTs increases, the average reward under all denoising configurations decreases noticeably. This is due to the increased UAV flight energy consumption caused by the larger number of GTs under a fixed computational capacity, leading to reduced energy efficiency.

\begin{figure}[!t]
\centering
\includegraphics[width=0.95\linewidth]{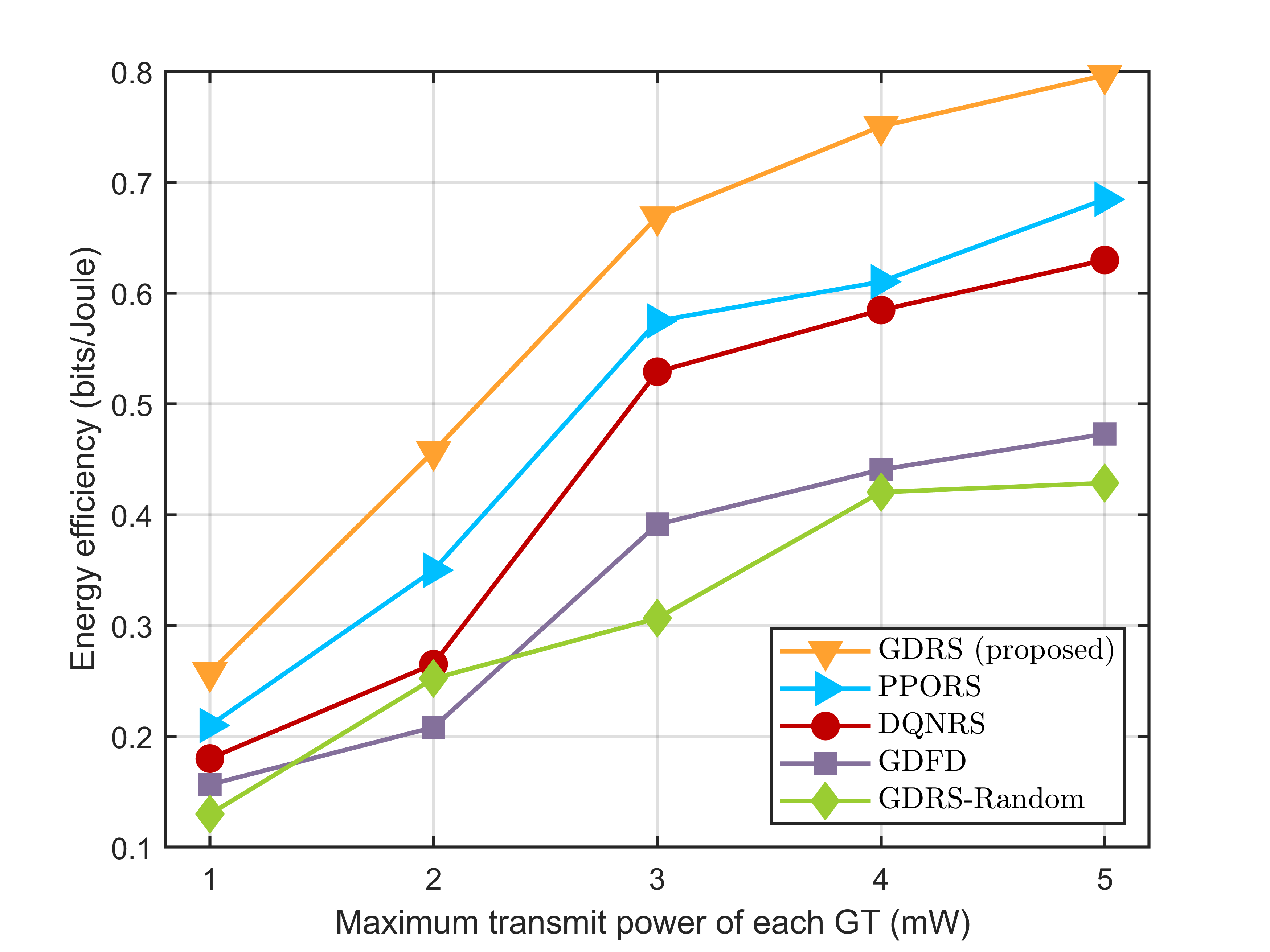}
\caption{Energy efficiency of different methods versus the maximum of transmit power of each GT.}
\label{EE_vs_transmit_power}
\end{figure}

Fig.~\ref{EE_vs_transmit_power} compares the energy efficiency of different methods versus the maximum transmit power of each GT. We find that all methods exhibit an upward trend as the maximum transmit power increases. This indicates that granting GTs higher transmit power generally enables more data bits to be offloaded and processed, and the growth in data bits outpaces the increase in UAV energy consumption. Moreover, we observe that the proposed GDRS method consistently achieves the highest energy efficiency across the entire power range, and the performance gap between GDRS and the baseline methods widens as the number of GTs increases. Additionally, we can observe that GDRS outperforms GDFD in terms of energy efficiency. This is because GDRS enables flexible interference management through successive interference cancellation, allowing simultaneous uplink transmissions over the shared bandwidth resource, whereas the GDFD method averages the bandwidth allocation, leading to reduced spectral efficiency and limited offloading performance.

\begin{figure}[!t]
\centering
\includegraphics[width=0.95\linewidth]{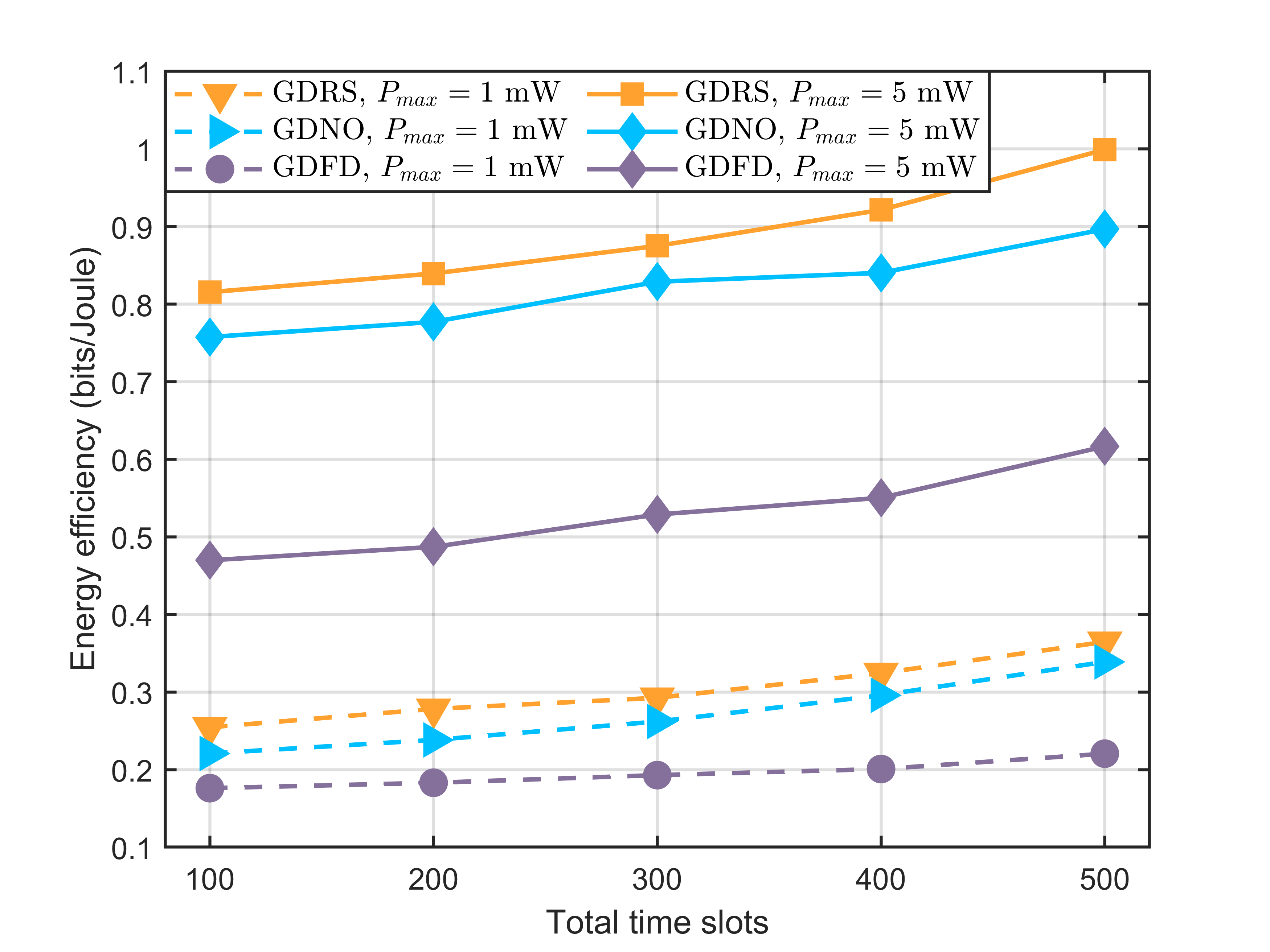}
\caption{Energy efficiency versus the total time slots of UAV under different multiple access methods.}
\label{EE_vs_total_time_slots}
\end{figure}

Fig.~\ref{EE_vs_total_time_slots} evaluates the energy efficiency versus the total number of time slots under different multiple access schemes. We compare the proposed RSMA-based system with representative NOMA-based scheme and FDMA-bassed scheme. First, all schemes exhibit an increase in energy efficiency as the time horizon expands. Intuitively, with more time slots, the UAV has greater flexibility to serve the offloading requests in a staggered manner and can fly in a more energy-conservative trajectory, thereby improving the overall energy efficiency. Besides, the RSMA-enabled system achieves the highest energy efficiency, significantly outperforming both the NOMA and OMA cases. As the number of time slots increases, the energy efficiency gap widens because RSMA can more fully exploit the extra time to optimize power allocation and decoding order among users. Since it still allows concurrent transmissions with SIC, the NOMA-based system is generally second-best here, but it suffers from higher interference and less flexible rate allocation than RSMA, resulting in lower efficiency. The OMA method performs the worst energy-wise due to its poor spectral utilization. By enabling part of the interference to be treated as useful information, RSMA strikes an advantageous balance that NOMA and OMA cannot, thus delivering better performance especially in high-load or extended-duration scenarios.

\begin{figure}[!t]
\centering
\includegraphics[width=0.95\linewidth]{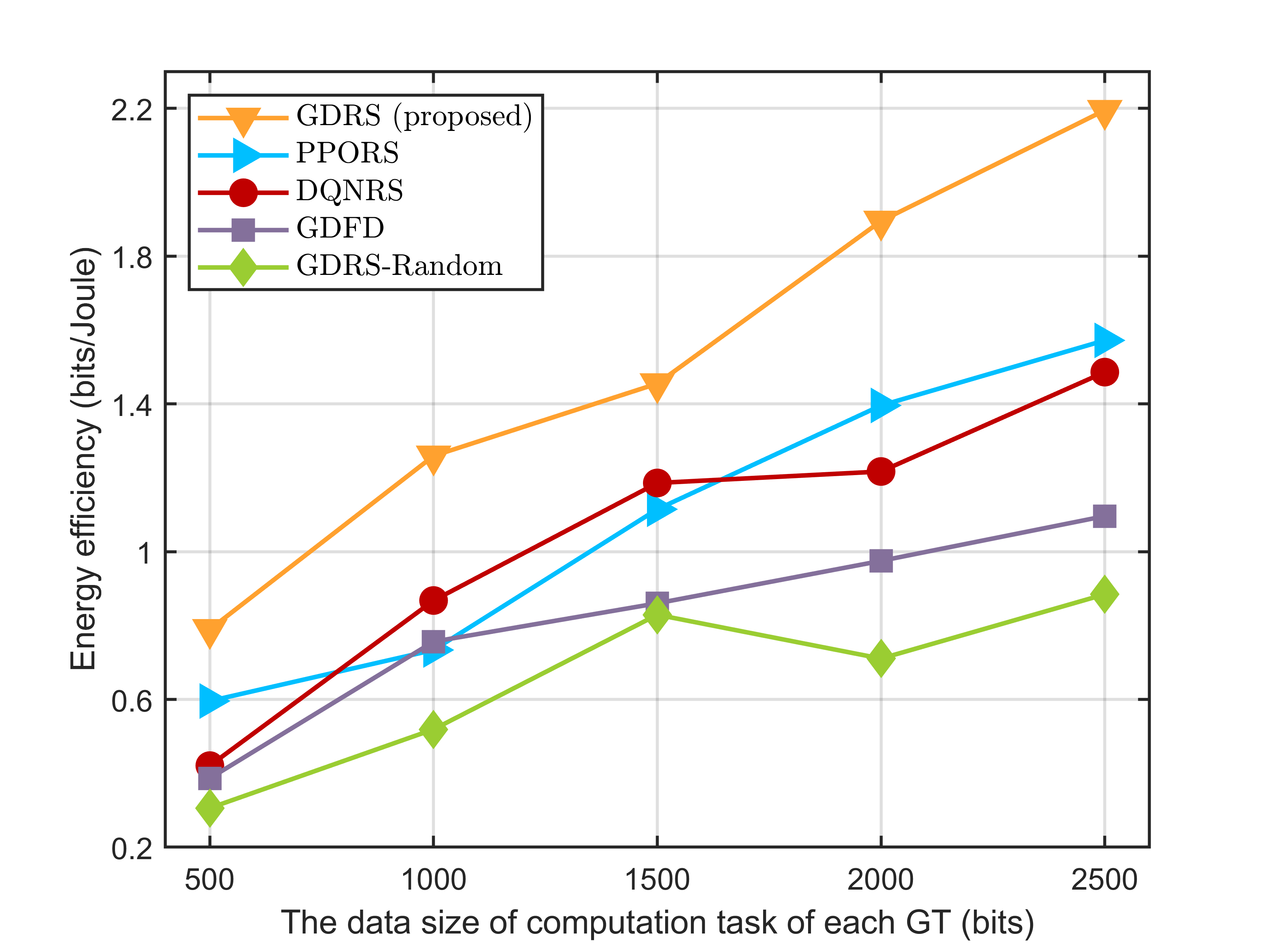}
\caption{Energy efficiency of different methods versus the data size of computation task of each GT.}
\label{EE_vs_data}
\end{figure}

In Fig.~\ref{EE_vs_data}, we analyze how the energy efficiency of the various methods varies with the size of computation task of GT. The plotted results indicate a general upward trend in energy efficiency as the task size grows. When tasks are very small, the overhead of coordination and UAV operation is relatively significant, leading to lower energy efficiency for all schemes. In this regime, GDRS still maintains a slight edge over other methods, but overall efficiency is limited by fixed costs. As the task size increases, all methods become more energy-efficient. This is because larger batches of data can be offloaded and processed in one go, amortizing the energy cost of UAV trajectory and link setup over more bits. Notably, the energy efficiency of GDRS method improves the most rapidly and to the highest values. And the margin between GDRS and the others widens with task size, indicating that GDRS scales especially well to heavy workloads. This is because GDRS method minimizes unnecessary UAV movement and efficiently multiplexes the uplink transmissions via RSMA.

\begin{figure}[!t]
\centering
\includegraphics[width=0.95\linewidth]{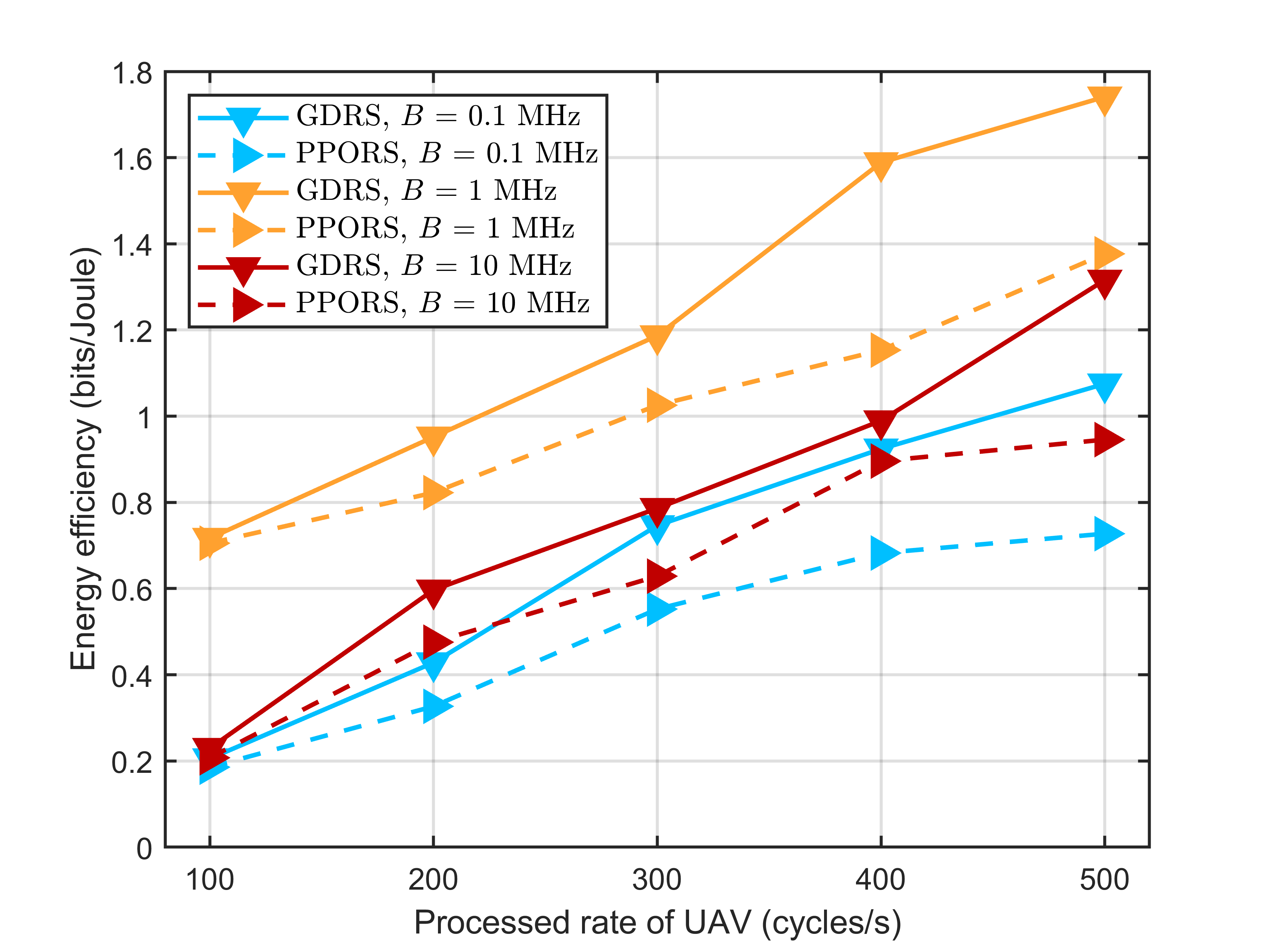}
\caption{Energy efficiency of GDRS and PPORS methods versus the processed rate of UAV under different communication bandwidth.}
\label{EE_vs_processed_data_of_UAV}
\end{figure}

Fig.~\ref{EE_vs_processed_data_of_UAV} plots the energy efficiency of GDRS versus the baseline PPORS against the processing rate of UAV under different bandwidth conditions. First, we observe that increasing the UAV's processing rate leads to higher energy efficiency for both GDRS and PPORS. With a faster onboard CPU, the UAV can execute offloaded tasks more quickly, reducing the time GTs spend transmitting. With a faster onboard CPU, the UAV can execute offloaded tasks more quickly, reducing the time GTs spend on data transmission. Moreover, the energy efficiency of GDRS consistently surpasses that of the PPORS method, demonstrating the advantage of the generative diffusion model in adapting to the complexity of low-altitude MEC offloading involving mobile UAV. In addition, under a fixed UAV processing rate, energy efficiency increases with communication bandwidth. When more bandwidth is available, uplink transmissions require less time or lower power for the same data volume, thereby improving the energy efficiency of both methods. Notably, even under constrained bandwidth conditions, GDRS outperforms the baseline. Specifically, when the UAV processing rate is set to $500$ cycles/s and the bandwidth to $0.1$ MHz, GDRS improves energy efficiency by approximately $47.3\%$ compared to PPORS, confirming that whether the bottleneck lies in computation or communication, GDRS achieves better energy utilization. 

\begin{figure}[!t]
\centering
\includegraphics[width=0.82\linewidth]{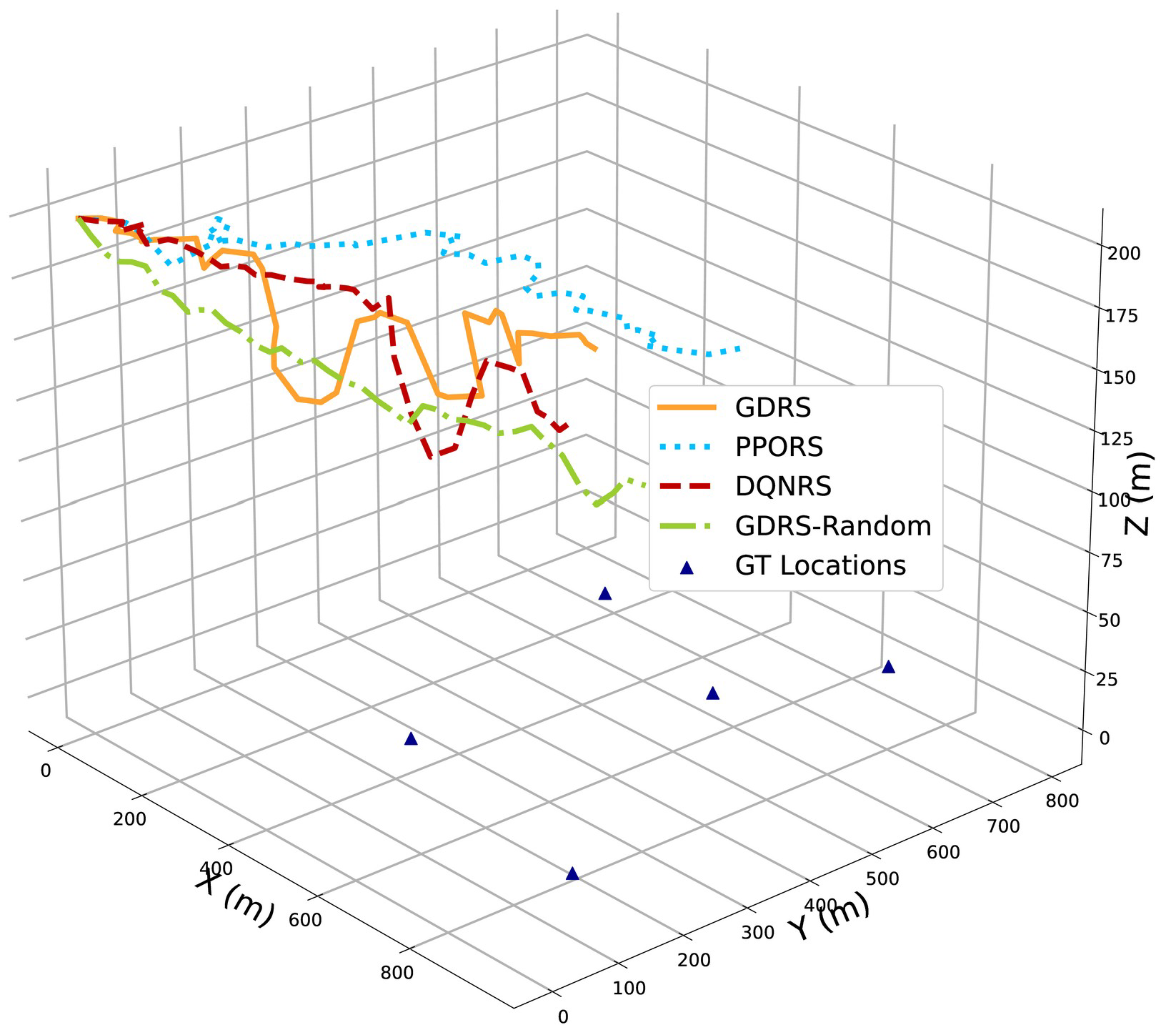}
\caption{UAV trajectories in 3D space for different methods.}
\label{UAV_trajectory}
\end{figure}

Fig.~\ref{UAV_trajectory} illustrates the UAV 3D trajectories under different methods. We observe that the proposed GDRS method yields a more structured and energy-aware trajectory that dynamically adapts to the spatial distribution of GTs, enabling efficient task offloading while minimizing redundant UAV motion. In contrast, the trajectories under PPORS and DQNRS exhibit erratic and less adaptive patterns, reflecting the limited exploration and poor generalization capability of traditional DRL approaches in high-dimensional hybrid action spaces. Additionally, the GDRS-Random method, which employs a stochastic RSMA decoding order, results in a degraded trajectory with inefficient hovering and detours, further highlighting the necessity of the decoding policy derived in Proposition~\ref{proposition_decoding_order}.

\section{Conclusion}\label{sec_conclusion}

In this paper, we investigated an uplink RSMA-enabled low-altitude MEC system to enhance energy efficiency under dynamic and interference-limited conditions. A novel generative diffusion model-enhanced DRL framework was proposed to jointly optimize UAV 3D trajectory, task offloading decisions, and power allocation. In addition, we analytically derive an optimal decoding order policy for the RSMA-enabled systems. Simulation results demonstrated that integrating RSMA with generative policy learning significantly improves energy efficiency and scalability in low-altitude MEC systems, particularly under dense user deployments and stringent resource constraints. Moreover, the diffusion model enhances policy exploration by generating structured actions that adapt to interference dynamics and user distributions, enabling more effective and energy-aware UAV behavior. However, this work assumes perfect knowledge of channel state information and GT locations, which may not be realistic in practical deployments. Therefore, future work will explore adaptive optimization schemes that are robust to imperfect channel conditions and mobile GT scenarios.



\bibliographystyle{IEEEtran}
\bibliography{references,references_r1}


\end{document}